\documentclass[12pt,reqno]{amsart}
\usepackage{graphicx}
\usepackage{amssymb,amsmath}
\usepackage{amsthm}
\usepackage{color}
\usepackage[pdf]{pstricks}
\usepackage{hyperref}
\usepackage{empheq}
\usepackage{pgfplots}

\numberwithin{equation}{section}

\usepackage[makeroom]{cancel}

\usepackage{amssymb}
\usepackage{epsfig}
\usepackage{extarrows}
\usepackage{amsfonts}
\usepackage{graphicx}
\usepackage{caption}
\usepackage{subcaption}
\usepackage{tikz}
\usepackage{color}

\hypersetup{hypertex=true,
	colorlinks=true,
	linkcolor=blue,
	anchorcolor=g,
	citecolor=red}

\newcommand{\R}{\mathbb{R}}

\DeclareMathOperator{\sn}{sn}
\DeclareMathOperator{\cn}{cn}
\DeclareMathOperator{\dn}{dn}

\newtheorem{remark}{Remark}[section]
\newtheorem{example}{Example}[section]
\newtheorem{proposition}{Proposition}[section]
\newtheorem{theorem}{Theorem}
\newtheorem{lemma}{Lemma}[section]
\newtheorem{corollary}{Corollary}

\begin{document}
	
\title[Kink breathers in the defocusing mKdV equation]{\bf Kink breathers on a traveling wave background \\in the defocusing modified Korteweg--de Vries equation}

\author{Lynnyngs Kelly Arruda}
\address[L. K. Arruda]{Departamento de Matem\'{a}tica, 
	Universidade Federal de S\~ao Carlos, S\~ao Carlos, S\~ao Paulo, 13 565 905, Brazil}
\email{lynnyngs@ufscar.br} 
	
\author{Dmitry E. Pelinovsky}
\address[D. E. Pelinovsky]{Department of Mathematics and Statistics, McMaster University, Hamilton, Ontario, Canada, L8S 4K1}
\email{pelinod@mcmaster.ca}

\begin{abstract}
	We characterize a general traveling periodic wave of the defocusing mKdV (modified Korteweg--de Vries) equation by using a quotient of products of Jacobi's elliptic theta functions. Compared to the standing periodic wave of the defocusing NLS (nonlinear Schr\"{o}dinger) equation, these solutions are special cases of Riemann's theta function of genus two. Based on our characterization, we derive a new two-parameter solution form which defines a general three-parameter solution form with the scaling transformation. Eigenfunctions of the Lax system for the general traveling periodic wave are also characterized as quotients of products of Jacobi's theta functions. As the main outcome of our analytical computations, we derive a new solution of the defocusing mKdV equation which describes the kink breather propagating on a general traveling wave background. 
\end{abstract}

\maketitle

\section{Introduction}

Dispersive shock waves (DSWs) arise in the wave dynamics on the infinite line from initial data with different constant boundary conditions at different infinities \cite{El,Hoefer}. Dynamics of DSWs is affected by the interaction with solitary waves and other localized perturbations \cite{Maiden,Mao,Sande,Spenger}.
Since DSWs are modeled as modulations of the traveling periodic waves, dynamics of solitary waves on the traveling periodic wave background have been recently studied for many integrable equations arising in the physics of fluids, optics, and plasmas \cite{ACEHL,SPP23}. 

Dynamics of the Korteweg--de Vries (KdV) equation has been considered in 
\cite{Bertola,Congy,Girotti,HMP23}, where it was found that the two basic propagations of solitary waves on the elliptic (cnoidal wave) background are represented by the bright (elevation) and dark (depression) profiles. Such time-periodic interactions of a spatially decaying wave and a spatially periodic wave are termed as {\em breathers} of the KdV equation \cite{HMP23}. Similar bright and dark breathers appear on the traveling wave background in the Benjamin--Ono equation \cite{ChenPel24}, where exact solutions are expressed in elementary functions compared to the elliptic functions. Extension of the breathers for the more general equations of the KP hierarchy can be found in \cite{Kakei,LiZhang}.

Propagation of solitary waves on the unstable elliptic background was studied for the focusing nonlinear Schr\"{o}dinger (NLS) equation in \cite{Biondini,Biondini2}, where rogue waves were shown to arise due to the modulational instability of both dnoidal and cnoidal waves \cite{CPW,Feng}. Similar propagations of solitary waves for the focusing modified KdV (mKdV) equation were studied in \cite{GravaSIMA,Grava,LingSAPM}. However, compared to the focusing NLS equation, the dnoidal wave is modulationally stable and supports stable propagation of bright breathers, whereas the cnoidal wave is modulationally unstable and supports rogue waves \cite{CP2018,CP2019}. 

For the defocusing NLS equation, the elliptic (snoidal wave) background is modulationally stable and dark breathers have been constructed in \cite{Ling,Shin,Takahashi}. Similarly, the snoidal wave is modulationally stable in the defocusing mKdV equation and dark breathers have been constructed in \cite{MP24}. It is rather interesting that the exact expressions of the dark breathers are different between the defocusing NLS and mKdV equations.  

{\em The purpose of this manuscript is to address the open problem arising in the construction of breathers in the defocusing mKdV equation posed in \cite{MP24}.}

Although the NLS and mKdV equations share the same spectral problem in the Lax system \cite{Ablowitz}, their traveling wave solutions are different. A general family of traveling wave solutions of the defocusing NLS equation is given by the elliptic functions which correspond to Riemann's theta function of genus one. These complex-valued solutions only give the snoidal wave of the defocusing mKdV equation since solutions of the mKdV equation are real-valued. However, the snoidal wave is not the most general traveling wave solution of the mKdV equation. The general solutions arise as the elliptic degeneration of Riemann's theta function of genus two. If the Lax spectrum of the snoidal wave contains only two bandgaps symmetrically relative to the origin, the Lax spectrum of the general traveling wave contains three bandgaps symmetrically relative to the origin. Dark breathers constructed in \cite{MP24} correspond to eigenvalues placed in the two bandgaps associated with the snoidal wave. It remained open in \cite{MP24} how to construct the kink breathers which correspond to eigenvalues placed in the central bandgap of the three bandgaps. Such kink breathers arise naturally in the previous numerical simulations of the defocusing mKdV equation \cite{Hoefer,Sande}. Kink breathers can be thought to represent heteroclinic connections between the traveling periodic wave of different polarities similar to the solutions constructed in \cite{Sprenger1,Sprenger2} in  non-integrable models.

Our approach to construct kink breathers is to express a general traveling wave as a quotient of products of Jacobi's elliptic theta functions, after which eigenfunctions of the Lax system are also expressed as quotients of products of Jacobi's theta functions. We show that zeros and poles of the first factorization are found uniquely in terms of parameters of the traveling wave solutions. However, we also show that zeros and poles of the second factorization (for eigenfunctions of the Lax system) are not found uniquely for a general value of the spectral parameter. Nevertheless, they are found uniquely if the spectral parameter is at the origin and we use the explicit factorization to construct the particular (symmetric) kink breather in a closed form with all parameters explicitly expressed in terms of parameters of the traveling wave solutions. This construction gives a novel solution of the mKdV equation in the form which is useful for interpretation of the numerical and laboratory experiments modeled by the defocusing mKdV equation. 

The paper is organized as follows. The main mathematical results on the characterization of a general traveling elliptic wave, its eigenfunctions for the spectral parameter at the origin, and the kink breathers are presented in Section \ref{sec-2}. Technical details of the proof of the main results can be found in Section \ref{sec-3}, where we also review other expressions for the general traveling elliptic wave, and in Section \ref{sec-4}. The concluding Section \ref{sec-5} poses an open question on 
the general characterization of eigenfunctions for nonzero values of the spectral parameter. Appendix \ref{app-A} reviews the unique parameterization of eigenfunctions of the Lax system for the snoidal wave of genus one which was used in \cite{B,MP24,Shin,Takahashi}.

\section{Main results}
\label{sec-2}

We consider the defocusing mKdV equation in the normalized form
\begin{equation}
\label{mkdv}
u_t-6u^2u_x+u_{xxx}=0,
\end{equation}
where $(x,t)\in \mathbb{R}\times\mathbb{R}$ and $u=u(x,t) \in \R$. As is well-known since the pioneering work \cite{Ablowitz}, classical solutions of the mKdV equation (\ref{mkdv}) arise as a compatibility condition $\partial_x \partial_t \varphi = \partial_t \partial_x \varphi$ of the following Lax system of linear equations for the eigenfunction $\varphi \in C^2(\mathbb{R} \times \mathbb{R},\mathbb{C}^2)$,
\begin{equation}
\label{LS}
\partial_x \varphi = U(\zeta,u) \varphi, \quad  
\partial_t \varphi = V(\zeta,u) \varphi, 
\end{equation}
where 
\begin{align*}
U(\zeta,u) &=\left(\begin{array}{ll} i\zeta & u\\ u & -i\zeta\end{array} \right), \\
V(\zeta,u)  &= \left( \begin{array}{ll} 4i\zeta^3+2i\zeta u^2 & 4\zeta^2u-2i\zeta u_x+2u^3-u_{xx}\\ 4\zeta^2u+2i\zeta u_x+2u^3-u_{xx} & -4i\zeta^3-2i\zeta u^2 \end{array}\right),
\end{align*}
and $\zeta \in \mathbb{C}$ is the spectral parameter. In what follows, we present the main results on the characterization of the general traveling periodic wave of the mKdV equation (\ref{mkdv}), eigenfunctions 
of the Lax system (\ref{LS}) for $\zeta = 0$, and the kink breathers. We shall use the normalized Jacobi's elliptic theta functions:
 \begin{align*}
\left\{ \begin{array}{l} 
\theta_1(y) = 2 \sum\limits_{n=1}^{\infty} (-1)^{n-1} q^{(n-\frac{1}{2})^2} \sin(2n-1) y, \\
\theta_4(y) = 1 + 2 \sum\limits_{n=1}^{\infty} (-1)^{n} q^{n^2} \cos 2n y,
\end{array} \right.
\end{align*}
where $q := e^{-\frac{\pi K'(k)}{K(k)}}$ with $K(k)$ being the complete elliptic integral and $K'(k) = K(k')$ with $k' = \sqrt{1-k^2}$. See \cite{A1990,B,BF1971,GR,Lawden,W} for review of elliptic functions. 
We use notations
\begin{equation}
\label{Jacobi-theta}
H(x) = \theta_1\left(\frac{\pi x}{2 K(k)} \right), \quad \Theta(x) = \theta_4\left(\frac{\pi x}{2 K(k)}\right), 
\end{equation}
and drop the dependence of the elliptic functions on $k \in (0,1)$ if it does not cause a confusion.

\subsection{Characterization of the general traveling periodic wave}

Travelling waves of the mKdV equation (\ref{mkdv}) are written in the form $u(x,t)=\phi(x+ct)$, where the real-valued profile $\phi$ satisfies the third-order equation 
\begin{equation}
\label{first}
    \phi'''-6\phi^2\phi'+c\phi'=0.
\end{equation}
Integration of (\ref{first}) yields the second-order equation 
\begin{equation}
\phi'' - 2 \phi^3 + c \phi = b, 
\label{second}
\end{equation}
where $b$ is the integration constant. Multiplying (\ref{second}) by $\phi'$ and integrating, we obtain the first-order invariant in the form: 
\begin{equation}
\label{third}
( \phi')^2 = Q(\phi), \qquad Q(\phi) := \phi^4 - c\phi^2 + 2b \phi + 2 d,
\end{equation}
where $d$ is another integration constant. We note that  equations (\ref{first}), (\ref{second}), and (\ref{third}) are written for $\phi = \phi(x)$, where $x$ stands for the traveling wave coordinate $x + ct$.

\begin{remark}
	\label{rem-scaling}
The following scaling transformation $(\tilde{\phi};\tilde{b},\tilde{c},\tilde{d}) \mapsto (\phi;b,c,d)$ given by 
\begin{equation}
\label{scal-transf}
\phi(x) = a \tilde{\phi}(a x), \quad c = a^2 \tilde{c}, \quad 
b = a^3 \tilde{b}, \quad d = a^4 \tilde{d},
\end{equation}
leaves solutions of the system (\ref{first}), (\ref{second}), and (\ref{third}) invariant with an arbitrary parameter $a \in \mathbb{R}$. This transformation can be used to reduce the number of independent parameters $(b,c,d)$ by one. 
\end{remark}

\begin{example}
	\label{ex-snoidal}
If $b = 0$, there exists the following one-parameter family of the periodic solutions of the system (\ref{first}), (\ref{second}), and (\ref{third})
\begin{equation}
\label{sn_potential}
\phi(x) = k\sn(x,k),\quad c = 1+k^2, \quad b = 0, \quad d = \frac{1}{2} k^2,
\end{equation}	
generated by the elliptic modulus $k \in (0,1)$. A two-parameter family of solutions of the system (\ref{first}), (\ref{second}), and (\ref{third}) for $b = 0$ is obtained with the scaling transformation (\ref{scal-transf}).
\end{example}

The first main result of this work is to represent the general periodic solution of the system (\ref{first}), (\ref{second}), and (\ref{third}) for $b \neq 0$ in the analytical form which involves the quotient of the product of Jacobi elliptic functions (\ref{Jacobi-theta}). 

\begin{theorem}
	\label{theorem-wave}
	Bounded periodic (nonconstant) solutions of the second-order equation (\ref{second}) exist if and only if $(b,c) \in \Omega$, where $\Omega$ is given by 
\begin{equation}
\label{domain}
\Omega = \{ (b,c) : \quad b \in (-b_c,b_c), \;\; c > 0\}, \qquad 
b_c := \frac{\sqrt{2 c^3}}{\sqrt{27}}.
\end{equation}	
These solutions for $b \in (0,b_c)$ and $c > 0$ can be uniquely parameterized by the real parameters $(\zeta_1,\zeta_2,\zeta_3)$ satisfying $0 < \zeta_3 < \zeta_2 < \zeta_1$ with 
\begin{equation}
\label{parameterization}
\begin{cases}
b = 4 \zeta_1 \zeta_2 \zeta_3, \\
c = 2(\zeta_1^2 + \zeta_2^2 + \zeta_3^2), \\
d = \frac{1}{2}(\zeta_1^4 + \zeta_2^4 + \zeta_3^4) - \zeta_1^2 \zeta_2^2 - \zeta_1^2 \zeta_3^2 - \zeta_2^2 \zeta_3^2.
\end{cases}
\end{equation}
If $\zeta_1 \neq \zeta_2 + \zeta_3$, the periodic profile $\phi$ is given explicitly by 
\begin{equation}
\label{gen-theta}
\phi(x) = (\zeta_1 - \zeta_2 - \zeta_3)  \frac{H(\nu x-\beta) H(\nu x + \beta)}{\Theta(\nu x-\alpha) \Theta(\nu x+\alpha)} \frac{\Theta^2(\alpha)}{H^2(\beta)}, 
\end{equation}
where $\nu > 0$, $k \in (0,1)$, $\alpha \in (0,K)$, and $\beta \in (0,K) \times i (0,K')$ are uniquely expressed by 
\begin{equation}
\label{parameters-nu-k-again}
\nu = \sqrt{\zeta_1^2 - \zeta_3^2},  \quad
k = \sqrt{\dfrac{\zeta_1^2 - \zeta_2^2}{\zeta_1^2 - \zeta_3^2}}, 
\end{equation}
and 
\begin{equation}
\label{parameters-nu-k-alpha}
\sn(\alpha) = \sqrt{\dfrac{\zeta_1 - \zeta_3}{\zeta_1 + \zeta_2}}, \quad 
\sn(\beta) = \sqrt{\frac{(\zeta_1 + \zeta_3)(\zeta_1 - \zeta_2 - \zeta_3)}{
		(\zeta_1 - \zeta_2) (\zeta_1 + \zeta_2 + \zeta_3)}}, 
\end{equation}
such that $\beta$ is real for $\zeta_1 > \zeta_2 + \zeta_3$ and purely imaginary for $\zeta_1 < \zeta_2 + \zeta_3$. If $\zeta_1 = \zeta_2 + \zeta_3$, the periodic profile $\phi$ is given by 
\begin{equation}
\label{wave-deg}
\zeta_1 = \zeta_2 + \zeta_3: \qquad 
\phi(x) = \frac{2 (\zeta_2 + \zeta_3) \zeta_3 \sn^2(\nu x) \Theta^2(\nu x) \Theta^2(\alpha)}{(\zeta_2 + 2 \zeta_3) \Theta(\nu x- \alpha) \Theta(\nu x + \alpha) \Theta^2(0)},	
\end{equation}
with the same $\nu > 0$, $k \in (0,1)$, and $\alpha \in (0,K)$.
\end{theorem}

Figure \ref{fig-0} displays the existence domain $\Omega$ in (\ref{domain}) by the green shaded area between the two boundaries (blue solid lines). The periodic solution in Example \ref{ex-snoidal} corresponds to the line $b = 0$, $c > 0$ (red dotted line). The periodic solution in Theorem \ref{theorem-wave} corresponds to the region in the upper half-plane between the blue solid line and the red dotted line. 

\begin{figure}[htb!]
	\includegraphics[width=10cm,height=8cm]{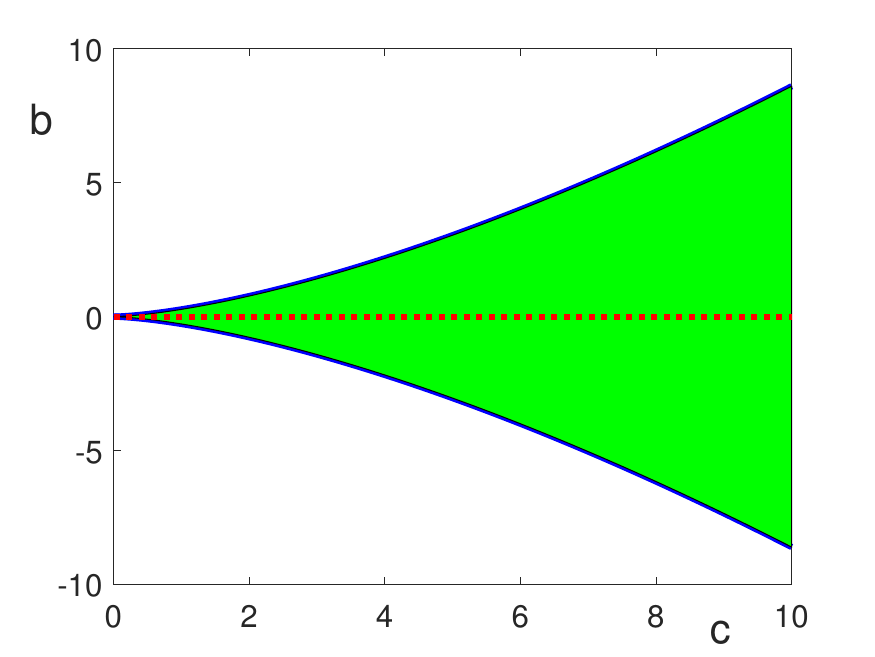}
	\caption{Domain $\Omega$ (green shaded area) between the two boundaries  (blue solid lines). Example \ref{ex-snoidal} corresponds to the red dotted line. }
	\label{fig-0}
\end{figure}

Figure \ref{fig-3} shows profiles $\phi$ of the periodic solutions of Theorem \ref{theorem-wave} for two choices of parameters $(\zeta_1,\zeta_2,\zeta_3)$ with  $\zeta_1 > \zeta_2 + \zeta_3$ (left) and $\zeta_1 < \zeta_2 + \zeta_3$ (right). Red dots show zeros of $\phi$ at $x = \pm \nu^{-1} \beta$ if $\beta \in (0,K)$ is real. The solution is positive-definite for $b \in (0,b_c)$ if $\beta = i (0,K')$ is purely imaginary.

\begin{remark}
	Since $H(x) = \sqrt{k} {\rm sn}(x) \Theta(x)$, see  \cite[(2.1.1)]{Lawden}, the snoidal wave (\ref{sn_potential}) can be rewritten as the following quotient:
	\begin{equation}
	\label{sn-theta}
	\phi(x) = \sqrt{k} \frac{H(x)}{\Theta(x)}.
	\end{equation}
	This solution is not expressed as a limiting case of the general solution (\ref{gen-theta}) as $\zeta_3 \to 0$ but can be derived from (\ref{parameterization}), (\ref{gen-theta}), (\ref{parameters-nu-k-again}), and (\ref{parameters-nu-k-alpha}) with $\zeta_3 = 0$ by using the Landen transformation \cite{Lawden} which shows that the definition of elliptic modulus $k$ is different between (\ref{gen-theta}) and (\ref{sn-theta}), see Example \ref{ex-particular-sol}. This illustrates that the snoidal wave (\ref{sn_potential}) is a particular case of Riemann's theta function of genus one, whereas the general elliptic wave (\ref{gen-theta}) with $0 < \zeta_3 < \zeta_2 < \zeta_1$ is a particular case of Riemann's theta function of genus two. 
	\label{rem-genus-two}
\end{remark}

\begin{figure}[htb!]
	\includegraphics[width=7.5cm,height=5cm]{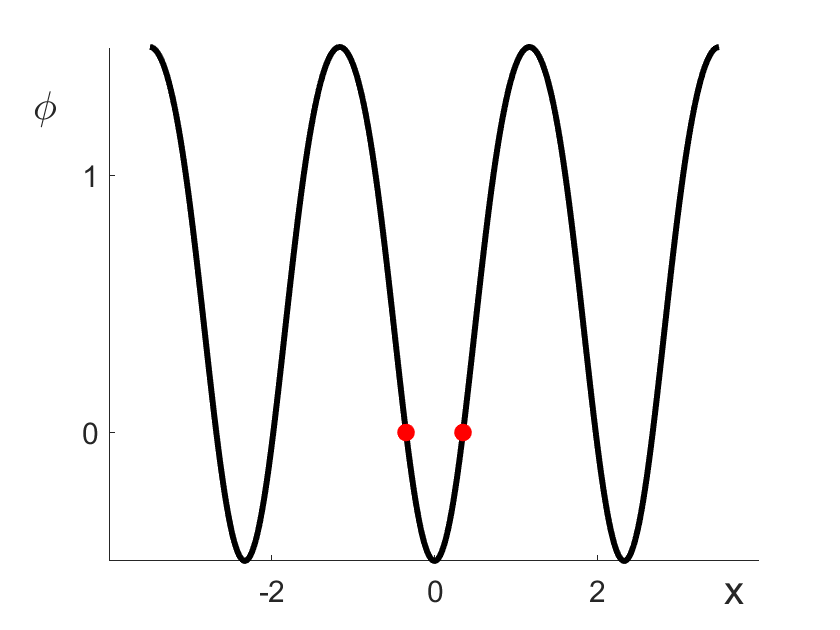}
	\includegraphics[width=7.5cm,height=5cm]{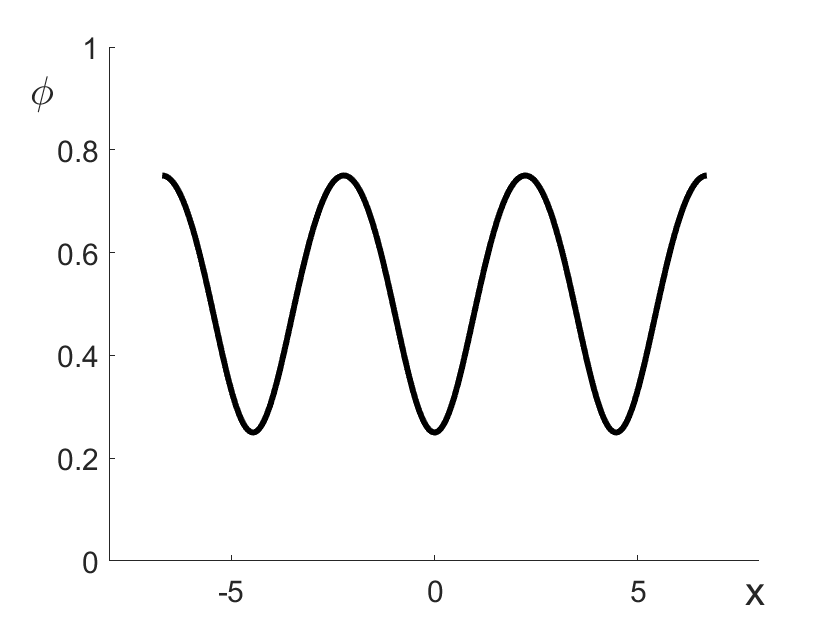}
	\caption{Profiles of $\phi$ versus $x$ for $\zeta_1 = 2$, $\zeta_2 = 1$,  $\zeta_3 = 0.5$ (left) and  $\zeta_1 = 1$, $\zeta_2 = 0.75$,  $\zeta_3 = 0.5$ (right).}
	\label{fig-3}
\end{figure}

Figure \ref{fig-1} displays the Lax spectrum for the snoidal wave (\ref{sn_potential}) (left) and for the general elliptic wave (\ref{gen-theta}) (right). The Lax spectrum is defined as the admissible set of the spectral parameter $\zeta$ for the Lax system (\ref{LS}) for which the eigenfunction $\varphi = \varphi(x,t)$ are bounded functions of $x$ on $\mathbb{R}$ for every $t \in \mathbb{R}$. The Lax spectrum for the defocusing mKdV equation (\ref{mkdv}) is a subset of $\mathbb{R}$. The two-gap spectrum on the left panel versus the three-gap spectrum on the right panel illustrates Remark \ref{rem-genus-two} on the difference between (\ref{sn_potential}) and (\ref{gen-theta}) as the genus-one and genus-two elliptic potentials, respectively.

\begin{figure}[htb!]
\includegraphics[width=7.5cm,height=6cm]{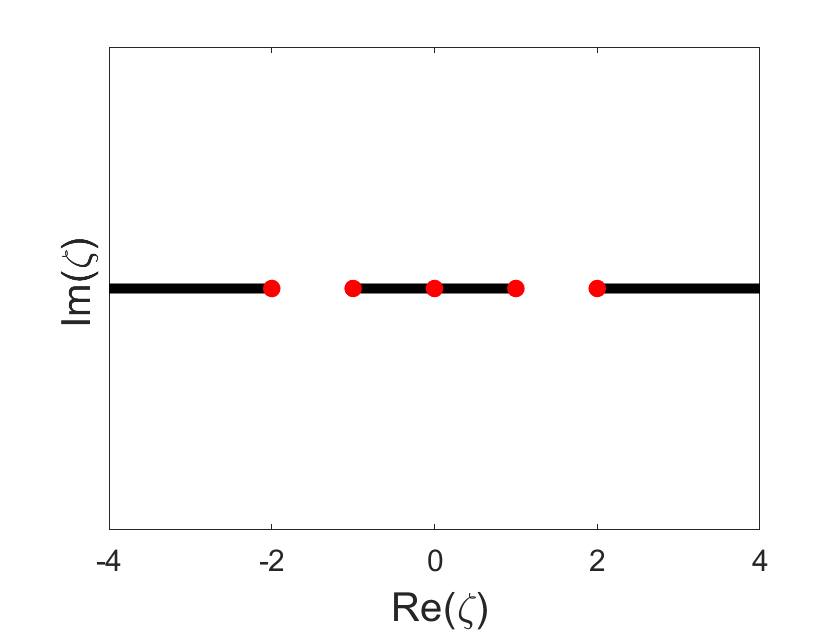}
\includegraphics[width=7.5cm,height=6cm]{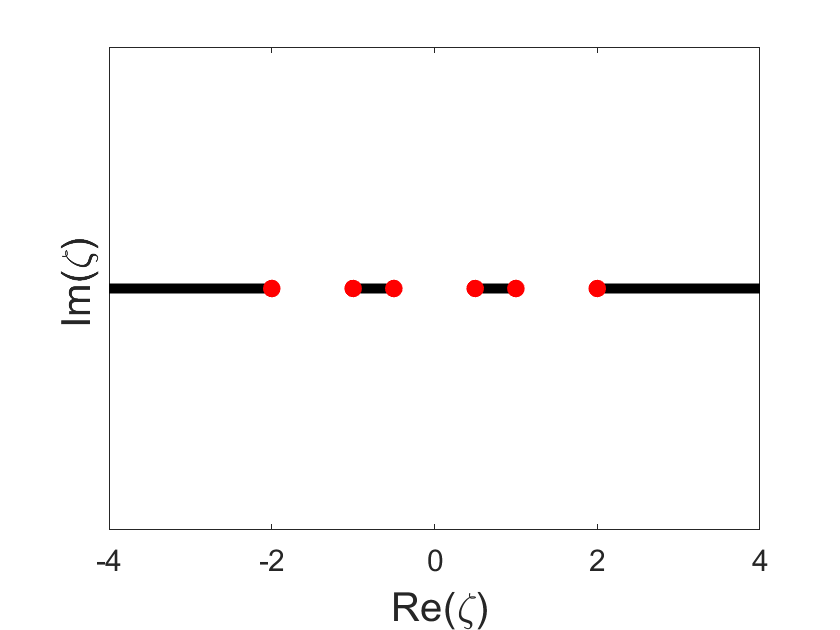}
\caption{Lax spectrum for $\zeta_1 = 2$, $\zeta_2 = 1$, and $\zeta_3 = 0$ (left) and for $\zeta_1 = 2$, $\zeta_2 = 1$, and $\zeta_3 = 0.5$ (right). Red dots show location of $\pm \zeta_1$, $\pm \zeta_2$, and $\pm \zeta_3$. }
\label{fig-1}
\end{figure}

\begin{example}
	\label{ex-hyperbolic-theta}
	Bounded periodic (nonconstant) solutions of Theorem \ref{theorem-wave} are defined for a fixed $\zeta_3 > 0$ and $\zeta_2 \in (\zeta_3,\zeta_1)$. If $\zeta_1 = \zeta_2$, then $k = 0$ and the solution becomes constant: $\phi(x) = \zeta_3$. If $\zeta_2 = \zeta_3$, then $k = 1$ and the solution becomes non-periodic (homoclinic) on the infinite line since $\lim\limits_{k \to 1} K(k) = \infty$. 
	The solution form (\ref{gen-theta}) for $k = 1$ is expressed in terms of the hyperbolic functions as 
\begin{equation}
\label{gen-theta-hyperbolic}
\phi(x) = (\zeta_1 - 2\zeta_2)  \frac{\sinh(\nu x-\beta) \sinh(\nu x + \beta)}{\cosh(\nu x-\alpha) \cosh(\nu x+\alpha)} \frac{\cosh^2(\alpha)}{\sinh^2(\beta)}, 
\end{equation}
where $\nu > 0$, $\alpha > 0$, and $\beta \in (0,\infty) \times i \left(0,\frac{\pi}{2}\right)$ are uniquely expressed by 
\begin{equation}
\label{parameters-nu-k-alpha-hyperbolic}
\nu = \sqrt{\zeta_1^2 - \zeta_2^2},  \quad
\tanh(\alpha) = \sqrt{\dfrac{\zeta_1 - \zeta_2}{\zeta_1 + \zeta_2}},  \quad 
\tanh(\beta) = \sqrt{\dfrac{(\zeta_1 + \zeta_2)(\zeta_1 - 2 \zeta_2)}{(\zeta_1 - \zeta_2)(\zeta_1 + 2 \zeta_2)}},
\end{equation}	
such that $\beta$ is real for $\zeta_1 > 2 \zeta_2$ and purely imaginary for $\zeta_1 < 2 \zeta_2$. The Lax specrum for the hyperbolic solution (\ref{gen-theta-hyperbolic}) corresponds to $(-\infty,-\zeta_1] \cup \{ -\zeta_2 \} \cup \{ \zeta_2\} \cup [\zeta_1,\infty)$.
\end{example}

Various solution forms of the traveling periodic waves have been used in the literature on the defocusing mKdV equation (\ref{mkdv}). One popular parameterizaton of solutions, e.g., used in \cite{Hoefer,Kam4}, is given by 
\begin{equation}
\label{form-2-intro}
\phi(x) = \dfrac{2 (\zeta_1 + \zeta_3)(\zeta_2 + \zeta_3)}{(\zeta_1 + \zeta_3) -(\zeta_1 - \zeta_2){\rm sn}^2(\nu x)} - \zeta_1 - \zeta_2 - \zeta_3,
\end{equation}
where $(\zeta_1,\zeta_2,\zeta_3)$ are the same as in (\ref{parameterization}) and $(\nu,k)$ are the same as in (\ref{parameters-nu-k-again}). However, as we show in Example \ref{ex-best-choice}, the constraint $\nu = 1$ suggested by the scaling transformation (\ref{scal-transf}) leads to non-unique choice of $(\zeta_1,\zeta_2,\zeta_3)$ for a given point $(b,c)$ in the existence region $\Omega$.

As a corollary of Theorem \ref{theorem-wave}, we present a novel  parameterization of the periodic profile $\phi$ which is independent on whether $\zeta_1 > \zeta_2 + \zeta_3$ or $\zeta_1 < \zeta_2 + \zeta_3$. This solution form is generated by two arbitrary parameters $\alpha \in (0,K)$ and $k \in (0,1)$, whereas the third arbitrary parameter $\zeta_1 \in (0,\infty)$ is generated by the scaling transformation (\ref{scal-transf}). The novel two-parameter representation gives all periodic solutions of the system (\ref{first}), (\ref{second}), and (\ref{third}) for $b \in (0,b_c)$ and $c > 0$ in the same way as the one-parameter solution form (\ref{sn_potential}) in Example \ref{ex-snoidal} gives all periodic solutions of the system (\ref{first}), (\ref{second}), and (\ref{third}) with $b = 0$. 

\begin{corollary}
	\label{cor-wave}
	For every $(\zeta_1,\zeta_2,\zeta_3) \in \R^3$ satisfying $\zeta_3 < \zeta_2 < \zeta_1$ in Theorem \ref{theorem-wave}, there exists a unique choice for $(\zeta_1,\alpha,k)$ in $(0,\infty) \times (0,K) \times (0,1)$ in the transformation 
		\begin{align}
		\label{transformation-nu}
\left\{ 
\begin{array}{l} 
\zeta_1 = \zeta_1, \\
\zeta_2 = \zeta_1 \dn(2 \alpha), \\ 
\zeta_3 = \zeta_1 \cn(2\alpha), \end{array}
\right.
\end{align}	
which also implies $\nu = \zeta_1 \sn(2\alpha)$ in (\ref{parameters-nu-k-again}). For $\zeta_1 = 1$, the elliptic profile $\phi$ in Theorem \ref{theorem-wave} can be uniquely expressed by 
	\begin{equation}
	\label{form-9}
	\phi(x) = (\dn(2\alpha) + \cn(2\alpha)) \frac{1 + k^2 \sn^2(\alpha) \sn^2(\sn(2\alpha) x)}{1 - k^2 \sn^2(\alpha) \sn^2(\sn(2\alpha) x)} - 1,
	\end{equation}	
	where $\alpha \in (0,K)$ and $k \in (0,1)$ are two arbirary parameters. 
\end{corollary}

\begin{proof}
By using the fundamental relations for the elliptic functions, 
\begin{equation}
\label{fund-elliptic}
\sn^2(z) + \cn^2(z) = 1, \qquad \dn^2(z) + k^2 \sn^2(z) = 1,
\end{equation} 
we obtain from the first formula in (\ref{parameters-nu-k-alpha}) that 
\begin{equation}
\label{elliptic-alpha}
\sn(\alpha) = \frac{\sqrt{\zeta_1 - \zeta_3}}{\sqrt{\zeta_1 + \zeta_2}}, \quad 
\cn(\alpha) = \frac{\sqrt{\zeta_2 + \zeta_3}}{\sqrt{\zeta_1 + \zeta_2}}, \quad 
\dn(\alpha) = \frac{\sqrt{\zeta_2 + \zeta_3}}{\sqrt{\zeta_1 + \zeta_3}}. 
\end{equation}	
	It follows from the addition formulas
\cite[(2.4.1)]{Lawden}:
\begin{equation}
\label{sn-addition}
\sn(u \pm v) = \frac{\sn(u) \cn(v) \dn(v) \pm \sn(v) \cn(u) \dn(u)}{1 - k^2 \sn^2(u) \sn^2(v)}
\end{equation}
that 
\begin{align}
\label{double-1}
\sn(2\alpha) = \frac{2 \sn(\alpha) \cn(\alpha) \dn(\alpha)}{1-k^2 \sn^4(\alpha)} = \frac{\sqrt{\zeta_1^2-\zeta_3^2}}{\zeta_1},
\end{align}
where we have used (\ref{elliptic-alpha}). By the fundamental relations (\ref{fund-elliptic}), we also obtain 
\begin{align}
\label{double-2}
\cn(2 \alpha) = \frac{\zeta_3}{\zeta_1}, \qquad 
\dn(2 \alpha) = \frac{\zeta_2}{\zeta_1},
\end{align}
which yield the transformation (\ref{transformation-nu}). The Jacobian of the transformation $(\zeta_1,\alpha,k) \mapsto (\zeta_1,\zeta_2,\zeta_3)$  is given by 
\begin{align*}
\frac{\partial (\zeta_1,\zeta_2,\zeta_3)}{\partial (\zeta_1,\alpha,k)} &=  \left| 
\begin{array}{ccc} 1 & 0 & 0 \\ 
\dn(2\alpha) & - 2k^2 \zeta_1 \sn(2\alpha) \cn(2\alpha) & \zeta_1 \partial_k \dn(2\alpha) \\ 
\cn(2\alpha) & - 2 \zeta_1 \sn(2\alpha) \dn(2\alpha) & \zeta_1 \partial_k \cn(2\alpha)
\end{array} \right|, \\
&= \zeta_1^2 \sn(2\alpha) \partial_k \left[ \dn^2(2\alpha) - k^2 \cn^2(2\alpha) \right] + 2k \zeta_1^2 \sn(2\alpha) \cn^2(2\alpha) \\
& 2k \zeta_1^2 \sn(2\alpha) [\cn^2(2 \alpha) - 1] \\
&= -2k \zeta_1^2 \sn^3(2\alpha),
\end{align*}
where we have used (\ref{fund-elliptic}). Since the Jacobian is nonzero for $\zeta_1 \in (0,\infty)$, $\alpha \in (0,K)$, and $k \in (0,1)$, the transformation is invertible and 
for every $(\zeta_1,\zeta_2,\zeta_3) \in \R^3$ satisfying $\zeta_3 < \zeta_2 < \zeta_1$, there exists a unique choice for $(\zeta_1,\alpha,k)$ in $(0,\infty) \times (0,K) \times (0,1)$.

Let us now set $\zeta_1 = 1$ and rewrite (\ref{form-2-intro}) in the equivalent form 
\begin{align*}
\phi(x) &= \frac{2 (1 + \cn(2\alpha))(\dn(2\alpha) + \cn(2\alpha))}{(1+\cn(2\alpha)) - (1-\dn(2\alpha)) \sn^2(\sn(2\alpha) x)} - 1 - \cn(2\alpha) - \dn(2\alpha) \\
&= (\dn(2\alpha) + \cn(2\alpha)) \frac{(1 + \cn(2\alpha)) + (1-\dn(2\alpha)) \sn^2(\sn(2\alpha) x)}{(1+\cn(2\alpha)) - (1-\dn(2\alpha)) \sn^2(\sn(2\alpha) x)} - 1.
\end{align*}
Due to definition of $\alpha$ in (\ref{parameters-nu-k-alpha}), we get 
$$
k^2 \sn^2(\alpha) = \frac{1 - \dn(2\alpha)}{1 + \cn(2\alpha)},
$$
from which we obtain (\ref{form-9}).
\end{proof}

\begin{example}
	\label{kink-hyperbolic}
For the hyperbolic solutions in Example \ref{ex-hyperbolic-theta}, we set $k = 1$ in Corollary \ref{cor-wave} and obtain $\zeta_1 = 1$, $\zeta_2 = \zeta_3 = {\rm sech}(2\alpha)$, and $\nu = \tanh(2\alpha)$. The solution form (\ref{form-9}) can be rewritten in the form 
	\begin{align}
	\phi(x) &= 2 {\rm sech}(2\alpha) \frac{1 + \tanh^2(\alpha) \tanh^2(z)}{1 - \tanh^2(\alpha) \tanh^2(z)} - 1 \notag \\
	&=  2 {\rm sech}(2\alpha) \frac{\cosh(2\alpha) \cosh(2z) + 1}{\cosh(2\alpha) + \cosh(2z)} - 1 \notag \\
	&= \frac{e^{2z} + e^{-2z} + 2 [1 - \sinh^2(2\alpha)] {\rm sech}(2\alpha)}{e^{2z} + e^{-2z} + 2 \cosh(2\alpha)},
	\label{soliton}
	\end{align}
	where $z = \tanh(2\alpha) x$.
\end{example}

\begin{remark}
	If $\zeta_3 = 0$, then it follows from (\ref{transformation-nu}) that $\alpha = \frac{1}{2} K$ and $\zeta_2 = \zeta_1 \sqrt{1-k^2}$. The expression (\ref{form-9}) is not similar to the snoidal solution (\ref{sn_potential}) because the elliptic modulus is different between the two expressions and the Landen transformation needs to be used, see Example \ref{ex-particular-sol}. 
\end{remark}

\subsection{Kink breathers}

If $\zeta = 0$, the Lax system of linear equations (\ref{LS}) with $u(x,t) = \phi(x+ct)$ admits two linearly independent solutions $\varphi = (p_0,q_0)^T$ and $\varphi = (p_0^*,q_0^*)^T$ given by 
\begin{equation}
\label{p-g-gen}
p_0(x,t) = q_0(x,t) = e^{\eta} 
\frac{\Theta(\nu \xi - \alpha)}{\Theta(\nu \xi + \alpha)}
\end{equation}
and 
\begin{equation}
\label{p-g-gen-second}
p_0^*(x,t) = - q_0^*(x,t) = e^{-\eta} 
\frac{\Theta(\nu \xi  + \alpha)}{\Theta(\nu \xi  - \alpha)},
\end{equation}
where $\xi = x + ct$ and $\eta = s_0 \xi - b t$ defined with 
$s_0 = \frac{\nu H'(2\alpha)}{H(2 \alpha)}$. 

The second main result of this work is to obtain the analytical solution of the mKdV equation (\ref{mkdv}) for the kink breather, which corresponds to a superposition of the kink soliton and the general traveling periodic wave of Theorem \ref{theorem-wave}.  The speed of the kink breather is defined in the coordinate $\eta = s_0 \xi - b t = s_0 (x + c_b t)$ with 
$$
c_b := c - \frac{b}{s_0}.
$$
If we choose $0 < \zeta_3 < \zeta_2 < \zeta_1$ as in Theorem \ref{theorem-wave}, 
then $b > 0$ and $s_0 > 0$ since it follows from (\ref{transformation-nu}) that $\alpha \in (0, \frac{1}{2}K)$ for $\zeta_3 > 0$. Hence, we have $c_b < c$, so that the kink moves to the right in the reference frame moving with the traveling wave in the coordinate $\xi = x + ct$.

By using one solution $u$ of the mKdV equation (\ref{mkdv}) and the eigenfunction $\varphi = (p,q)^T$ of the Lax system (\ref{LS}) with spectral parameter $\zeta$, we can construct another solution $\hat{u}$ of the same mKdV equation (\ref{mkdv}) from the one-fold Darboux transformation \cite{CP2018,CP2019,LingSAPM,MP24}: 
\begin{equation}
\label{DT}
\hat{u} = u - \frac{4 i \zeta p q}{p^2-q^2}.
\end{equation}
The kink breathers arises in the singular limit $\zeta \to 0$ of the Darboux transformation (\ref{DT}), for which $(p,q)$ is close to one of the two linearly independent solutions (\ref{p-g-gen}) and (\ref{p-g-gen-second}). By analyzing the solution in the limit $\zeta \to 0$, we have obtained the following theorem. 

\begin{theorem}
	\label{theorem-breather}
	Consider the traveling wave with the elliptic profile $\phi$ in Theorem \ref{theorem-wave} for $0 < \zeta_3 < \zeta_2 < \zeta_1$. A bounded kink breather solution of the mKdV equation (\ref{mkdv}) is given by 
\begin{equation}
\label{kink-breather}
u(x,t) = \frac{4 \phi(\xi) \Theta^2(\nu \xi + \alpha) + e^{2(\eta + \eta_0)} \Theta^2(\nu \xi - \alpha) (2 \phi(\xi) \phi'(\xi) - \phi''(\xi) - b)}{4 \Theta^2(\nu \xi + \alpha) + e^{2(\eta + \eta_0)} \Theta^2(\nu \xi - \alpha)  (c + 2 \phi'(\xi) - 2 \phi(\xi)^2)},
\end{equation}
where $\xi = x + ct$, $\eta = s_0(x + c_b t)$, and 
$\eta_0 \in \mathbb{R}$ is the arbitrary translational parameter. The kink breather is characterized by the breather speed
\begin{align}
\label{speed-breather}
c_b = 2(\zeta_1^2 + \zeta_2^2 +\zeta_3^2) - \frac{4 \zeta_1 \zeta_2\zeta_3}{\zeta_2 + \zeta_3 -\zeta_1 + 2 \sqrt{\zeta_1^2 - \zeta_3^2} Z(\alpha)} < c
\end{align}	
and the breather localization (inverse half-width) 
\begin{align}
\label{localization}
s_0 = \zeta_2 + \zeta_3 -\zeta_1 + 2 \sqrt{\zeta_1^2 - \zeta_3^2} Z(\alpha) > 0,
\end{align}
where $Z(x) := \frac{\Theta'(x)}{\Theta(x)}$ is Jacobi's zeta function. 
The asymptotic behavior of the kink breathers is given by the limits 
\begin{equation}
\label{limits-breather}
u(x,t) \to \left\{ \begin{array}{ll} 
\phi(\xi) \quad &\mbox{\rm as} \;\; \eta \to -\infty, \\
-\phi(\xi - 2 \nu^{-1} \alpha) \quad &\mbox{\rm as} \;\; \eta \to +\infty, 
\end{array} \right.
\end{equation}
so that $2 \nu^{-1} \alpha$ is the phase shift impaired by the kink breather 
in addition to the sign flip.
\end{theorem}

Figure \ref{fig-breathers} gives an illustrative example of the kink breather in Theorem \ref{theorem-breather} for $\zeta_1 = 1$, $k = 0.9$, and $\alpha = 0.25 K$. It follows from (\ref{transformation-nu}) and (\ref{localization}) that 
$\zeta_2 \approx 0.66$, $\zeta_3 \approx 0.55$, and $s_0 \approx 0.43$. Since $\xi = x + ct$ and $c_b < c$ as in (\ref{speed-breather}), the kink moves to the right in the reference frame moving with the periodic wave. 

\begin{figure}[htb!]
	\centering
	\includegraphics[width=7.5cm,height=5cm]{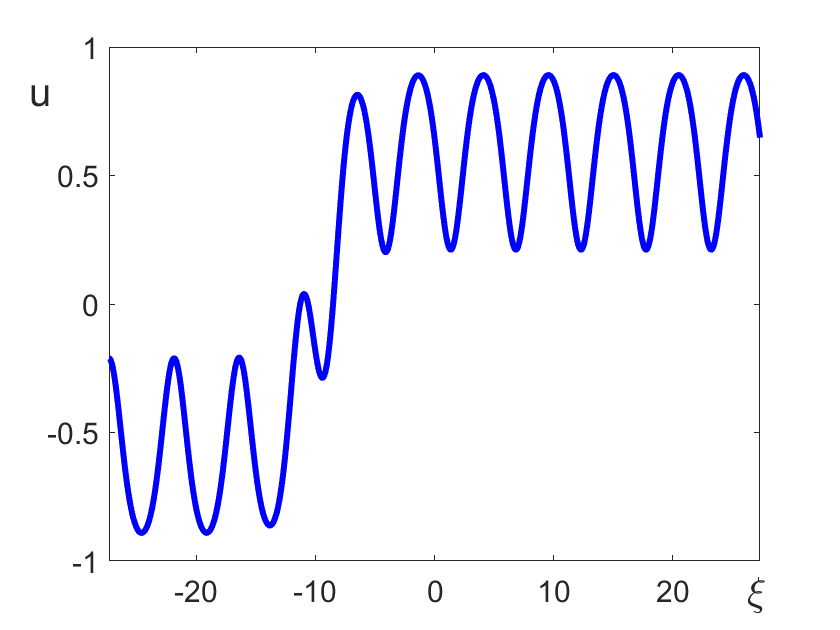}
	\includegraphics[width=7.5cm,height=5cm]{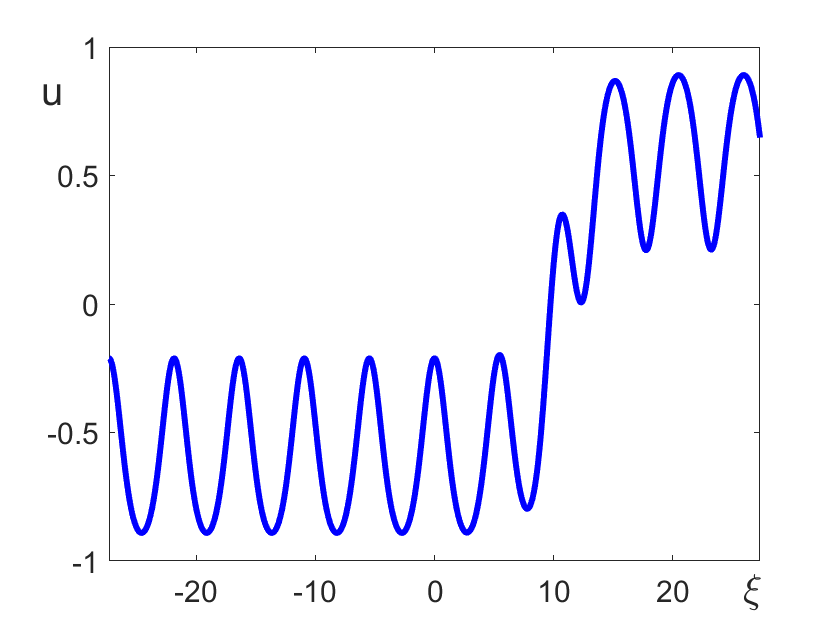}
	\caption{Plots of $u(x,t)$ for the kink breather solution (\ref{kink-breather}) 
		with $\zeta_1 = 1$, $k = 0.9$, and $\alpha = 0.25 K$ versus $\xi = x + ct$ for $t = -3$ (left) and $t = 3$ (right). The kink moves to the right relative to the periodic wave and flips its sign. The periodic wave impares the phase shift (\ref{limits-breather}) due to the interaction with the kink.}
	\label{fig-breathers}
\end{figure}

Next we show that the kink breather (\ref{kink-breather}) of Theorem \ref{theorem-breather} recovers as $k \to 1$ the two-soliton solution constructed in \cite{MP24}, where one soliton is the kink and the other soliton is the solution with the hyperbolic profile in Example \ref{kink-hyperbolic}. The result is given by the following corollary.

\begin{corollary}
	\label{cor-two-solitons} Consider  $\zeta_1 = 1$, $\zeta_2 = \zeta_3 = {\rm sech}(2\alpha)$, and $\nu = \tanh(2\alpha)$, where $\alpha \in (0,\infty)$ is the only parameter for the hyperbolic profile (\ref{soliton}). The kink breather of Theorem \ref{theorem-breather} for $k = 1$ is equivalent up to the translational parameters in $\xi$ and $\eta$ to the two-soliton solution 
	\begin{equation}
	\label{two-solitons}
	u(x,t) = 
	\frac{\sinh(\eta+2\alpha)  e^{-2z}  +  \sinh(\eta-2\alpha) e^{2z} + 2 
		\sinh(\eta) (1 - \sinh^2(2\alpha)) {\rm sech}(2\alpha)}{\cosh(\eta+2\alpha)  e^{-2z}  +  \cosh(\eta-2\alpha) e^{2z} + 2 
		\cosh(\eta) \cosh(2\alpha)},
	\end{equation}
	where $\xi = x + 2t + 4 {\rm sech}^2(2 \alpha) t$, $\eta = x + 2t$, and $z = \tanh(2\alpha) \xi$.
\end{corollary}

\begin{proof}
We compute from (\ref{speed-breather}) and (\ref{localization}) with $\zeta_1 = 1$, $\zeta_2 = \zeta_3 = {\rm sech}(2\alpha)$, and $\nu = \tanh(2\alpha)$ that  
\begin{align*}
s_0 &= -1 + 2 {\rm sech}^2(2\alpha) + 2 \tanh(2\alpha) \tanh(\alpha) \\
&= 1 - 2 {\rm sech}(2\alpha) {\rm sech}(\alpha) \left[ \cosh(2\alpha) \cosh(\alpha) - \cosh(\alpha) - \sinh(\alpha) \sinh(2\alpha) \right] \\
&= 1
\end{align*}
and 
\begin{align*}
c_b &= 2 (1 + 2 {\rm sech}^2(2\alpha)) - 4 {\rm sech}^2(2\alpha) = 2.
\end{align*}
This yields $\xi = x + 2t + 4 {\rm sech}^2(2 \alpha) t$ and $\eta = x + 2t$.
The solution form (\ref{two-solitons}) must be identical to the solution form  (\ref{kink-breather}) for $k = 1$ up to the choice of translations in $\xi$ and $\eta$ since both solutions correspond to the kink for a simple eigenvalue at $0$ and a soliton with the hyperbolic profile (\ref{soliton})  for a pair of simple eigenvalues $\pm {\rm sech}(2\alpha)$ in the gap of the continuous spectrum $(-\infty,-1] \cup [1,\infty)$ of the Lax system (\ref{LS}).
In other words, the Lax spectrum for the two-soliton solution (\ref{two-solitons}) is 
$$
(-\infty,-\zeta_1] \cup \{ -{\rm sech}(2\alpha)\} \cup \{ 0 \}
\cup \{ {\rm sech}(2\alpha)\} \cup [\zeta_1,\infty).
$$ 
To show the correspondence between (\ref{kink-breather}) with $k = 1$ and (\ref{two-solitons}), we first note that the expression (\ref{two-solitons}) converges to the shifted kink $\tanh(\eta \mp 2 \alpha)$ as $\xi \to \pm \infty$ and to the shifted hyperbolic profile (\ref{soliton}) with phase shifts $\pm \alpha \coth(2\alpha)$ along $x$ as $\eta \to \pm \infty$ and with the sign flip as $\eta \to -\infty$. 

The solution form (\ref{kink-breather}) can be rewritten in the equivalent form:
\begin{equation*}
u(x,t) = \phi(\xi) - \frac{2 \zeta_1 \zeta_2 \zeta_3  e^{2(\eta + \eta_0)} \Theta^2(\nu \xi - \alpha)}{\Theta^2(\nu \xi + \alpha) + e^{2(\eta + \eta_0)} \Theta^2(\nu \xi - \alpha)  (\zeta_1^2 \cn^2(\nu \xi - \alpha) + \zeta_2^2 \sn^2(\nu \xi - \alpha))},
\end{equation*}
where we have used equation (\ref{expression-1}) below. In the limit $k \to 1$, this expression 
with the normalization $\zeta_1 = 1$ yields
\begin{align}
u(x,t)  &=  2 {\rm sech}(2\alpha) \frac{\cosh(2\alpha) \cosh(2z) + 1}{\cosh(2\alpha) + \cosh(2z)} - 1  \notag \\
& \qquad - \frac{2 e^{2 (\eta + \eta_0)} \cosh^2(z - \alpha)}{\cosh^2(2\alpha) \cosh^2(z + \alpha) + e^{2 (\eta + \eta_0)}[\cosh^2(2\alpha) + \sinh^2(z - \alpha)]},
\label{two-soliton-new}
\end{align}
where we have used the second equality in (\ref{soliton}) with $z = \tanh(2\alpha) \xi$.
In order to reduce (\ref{two-soliton-new}) to (\ref{two-solitons}), 
we need to fix the translational parameters in the definitions of $\xi$ and $\eta$. It follows from comparison between (\ref{limits-breather}) as $\eta \to \pm \infty$ and (\ref{two-solitons}) as $\eta \to \pm \infty$, which is the shifted hyperbolic profile (\ref{soliton}) with phase shifts $\pm \alpha$ along $z$ as $\eta \to \pm \infty$ and with the sign flip as $\eta \to -\infty$, 
that $z$ in (\ref{two-soliton-new}) should be replaced by $z + \alpha$ 
and the sign should be flipped. This yields the new expression 
instead of (\ref{two-soliton-new}):
\begin{align}
u(x,t) &= 1 - 2 {\rm sech}(2\alpha) \frac{\cosh(2\alpha) \cosh(2z + 2 \alpha) + 1}{\cosh(2\alpha) + \cosh(2z + 2 \alpha)}  \notag \\
& \qquad + \frac{2 e^{2 (\eta + \eta_0)} \cosh^2(z)}{\cosh^2(2\alpha) \cosh^2(z + 2\alpha) + e^{2 (\eta + \eta_0)}[\cosh^2(2\alpha) + \sinh^2(z)]},
\label{two-soliton-shifted}
\end{align}
In the limit $z \to \pm \infty$, the expression (\ref{two-soliton-shifted}) yields
\begin{align*}
u(x,t) \to  -1 + \frac{2 e^{2 (\eta + \eta_0)}}{\cosh^2(2\alpha) e^{\pm 4 \alpha} + e^{2 (\eta + \eta_0)}},
\end{align*}
which must recover $\tanh(\eta \mp 2 \alpha)$, which follows from (\ref{two-solitons}) as $z \to \pm \infty$. This yields the definition of $\eta_0$ as $e^{2 \eta_0} = \cosh^2(2 \alpha)$. With this definition of $\eta_0$, we rewrite the expression  (\ref{two-soliton-shifted}) in the form 
\begin{align}
u(x,t) &=  1 - 2 {\rm sech}(2\alpha) \frac{\cosh(2\alpha) \cosh(2z + 2 \alpha) + 1}{\cosh(2\alpha) + \cosh(2z + 2 \alpha)} \notag \\
& \qquad + \frac{2 e^{2 \eta} \cosh^2(z)}{\cosh^2(z + 2\alpha) + e^{2 \eta}[\cosh^2(2\alpha) + \sinh^2(z)]}. \label{two-soliton-final}
\end{align}
To show the direct equivalence of (\ref{two-solitons}) and (\ref{two-soliton-final}), 
we note that 
\begin{align*}
\cosh(2\alpha) + \cosh(2z + 2 \alpha) = 2 \cosh(z + 2 \alpha) \cosh(z), 
\end{align*}
\begin{align*}
& \quad \cosh^2(z + 2\alpha) + e^{2 \eta}[\cosh^2(2\alpha) + \sinh^2(z)] \\
&= 
\cosh^2(z + 2\alpha) + \frac{1}{2} e^{2 \eta}[\cosh(4\alpha) + \cosh(2 z)] \\
&= \cosh(z + 2\alpha) [\cosh(z+2\alpha) + e^{2\eta} \cosh(z-2\alpha)], 
\end{align*}
and 
\begin{align*}
& \cosh(\eta+2\alpha)  e^{-2z}  +  \cosh(\eta-2\alpha) e^{2z} + 2 
\cosh(\eta) \cosh(2\alpha) \\
& \quad = e^{-\eta} [ \cosh(2z + 2 \alpha) + \cosh(2\alpha)] + e^{\eta} [\cosh(2z - 2 \alpha) + \cosh(2\alpha)] \\
& \quad = 2 \cosh(z) [e^{-\eta} \cosh(z+2\alpha) + e^{\eta} \cosh(z-2\alpha)].
\end{align*}
Subtracting (\ref{two-solitons}) from $1$ yields 
\begin{align*}
1 - u(x,t) &= \frac{e^{-\eta} [\cosh(2\alpha) \cosh(2z+2\alpha) + 1] + e^{\eta} \sinh^2(2\alpha)}{\cosh(2\alpha) \cosh(z) [e^{-\eta} \cosh(z+2\alpha) + e^{\eta} \cosh(z-2\alpha)]}.
\end{align*}
This expression is compared with the last two terms in (\ref{two-soliton-final}): 
\begin{align*}
\quad & 2 {\rm sech}(2\alpha) \frac{\cosh(2\alpha) \cosh(2z + 2 \alpha) + 1}{\cosh(2\alpha) + \cosh(2z + 2 \alpha)} - \frac{2 e^{2 \eta} \cosh^2(z)}{\cosh^2(z + 2\alpha) + e^{2 \eta}[\cosh^2(2\alpha) + \sinh^2(z)]} \\
&= \frac{\cosh(2\alpha) \cosh(2z + 2 \alpha) + 1}{\cosh(2\alpha) \cosh(z) \cosh(z+2\alpha)} - \frac{2 e^{\eta} \cosh^2(z)}{\cosh(z + 2\alpha) [e^{-\eta}  \cosh(z+2\alpha) + e^{\eta} \cosh(z-2\alpha)]} \\
&= \frac{e^{-\eta} [\cosh(2\alpha) \cosh(2z+2\alpha) + 1] + e^{\eta} \sinh^2(2\alpha)}{\cosh(2\alpha) \cosh(z) [e^{-\eta} \cosh(z+2\alpha) + e^{\eta} \cosh(z-2\alpha)]},
\end{align*}
where the simple expression for $e^{\eta} \sinh^2(2\alpha)$ in the numerator is obtained from  
\begin{align*}
& \quad \cosh(z-2\alpha) [ \cosh(2\alpha) \cosh(2z + 2 \alpha) + 1 ] - 2 \cosh(2\alpha) \cosh^3(z) \\
&= \cosh(z-2\alpha) [ \cosh(2\alpha) \cosh(2z + 2 \alpha) + 1 ] - \cosh^2(z) [\cosh(z + 2 \alpha) + \cosh(z-2\alpha)]\\
&= \cosh(z-2\alpha) [ \cosh(2\alpha) \cosh(2z + 2 \alpha) - \sinh^2(z) ] \\
& \qquad \qquad - \cosh(z + 2 \alpha) [\cosh^2(z) + \sinh^2(2\alpha) - \sinh^2(2\alpha)] \\
&= \cosh(z-2\alpha) [ \cosh(2\alpha) \cosh(2z + 2 \alpha) - \sinh^2(z) ] \\
& \qquad \qquad - \cosh(z + 2 \alpha) [\cosh(z+2\alpha) \cosh(z - 2\alpha) - \sinh^2(2\alpha)] \\
&= \cosh(z-2\alpha) [ \cosh(2\alpha) \cosh(2z + 2 \alpha) - \sinh^2(z) - \cosh^2(z + 2 \alpha)] + \cosh(z + 2 \alpha) \sinh^2(2\alpha) \\
&= \cosh(z + 2 \alpha) \sinh^2(2\alpha).
\end{align*}
This completes the proof of equivalence of (\ref{two-solitons}) and (\ref{two-soliton-final}).
\end{proof}

Figure \ref{fig-solitons} gives an illustrative example of the two-soliton solution in Corollary \ref{cor-two-solitons} for $\zeta_1 = 1$ and $\alpha = 0.5$. Since $\eta = x + 2 t$ and $2 = c_b < c = 2 + 4 {\rm sech}^2(2\alpha)$, the soliton with the hyperbolic profile (\ref{soliton}) moves to the left in the reference frame moving with the kink. The exact solution (\ref{two-solitons}) illustrated in Figure \ref{fig-solitons} matches well the numerical experiment  in Figure 11 in \cite{Sande} (with the opposite direction of the time variable used in \cite{Sande}), which displays the flip of the soliton across the kink.

\begin{figure}[htb!]
	\centering
	\includegraphics[width=5cm,height=5cm]{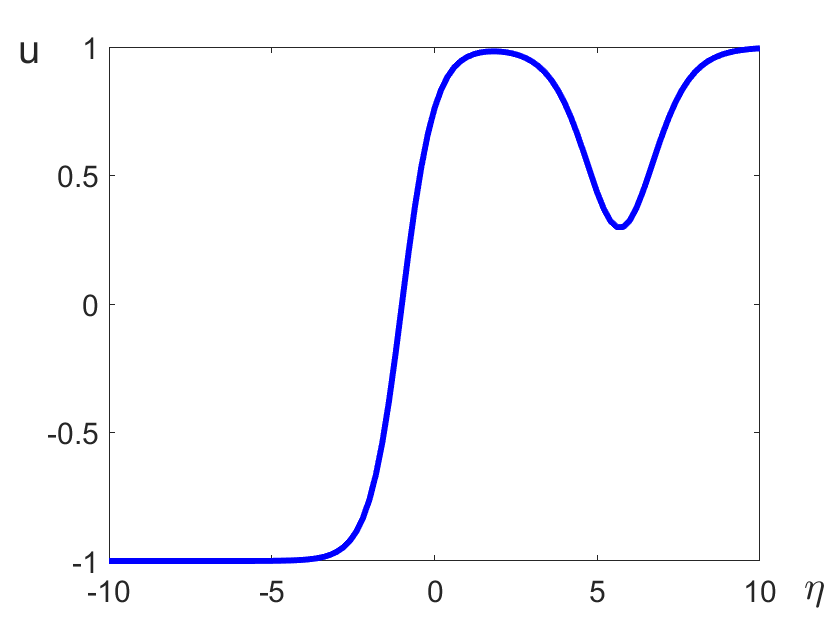}
	\includegraphics[width=5cm,height=5cm]{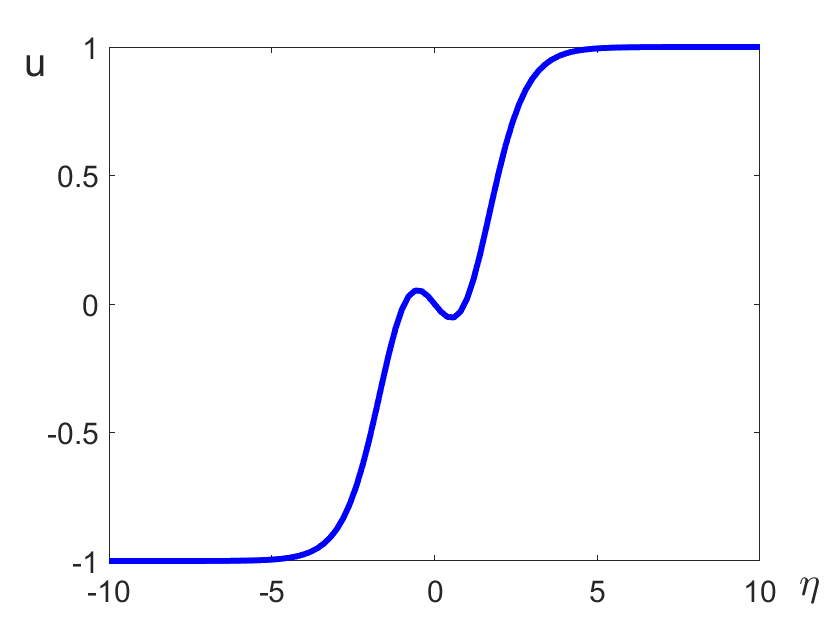}
	\includegraphics[width=5cm,height=5cm]{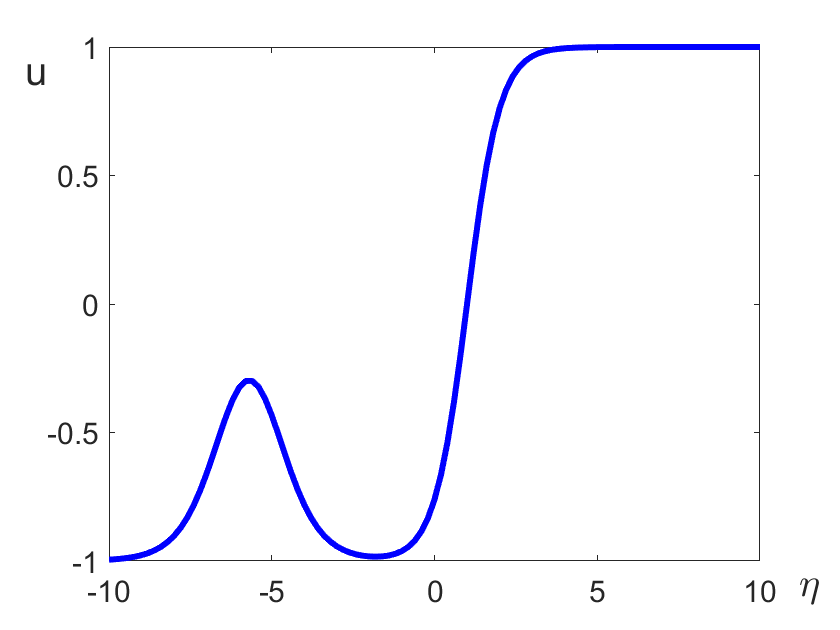}
	\caption{Plots of $u(x,t)$ for the two-soliton solution (\ref{two-solitons}) 
		with $\alpha = 0.5$ versus $\eta = x + 2t$ for $t = -3$ (left), $t = 0$ (middle), and $t = 3$ (right). The soliton with the hyperbolic profile (\ref{soliton}) moves to the left relative to the kink and flips the sign after the interaction with the kink. Both the soliton and the kink impares the phase shifts due to the interaction.}
	\label{fig-solitons}
\end{figure}

\section{Proof of Theorem \ref{theorem-wave}}
\label{sec-3}

In addition to the proof of Theorem \ref{theorem-wave}, we will review other solution forms for the general traveling periodic wave of the mKdV equation (\ref{mkdv}) used in the literature. 

\subsection{Solution form parameterized by roots of $Q$}

Let us factorize the polynomial $Q$ in (\ref{third}) by its four roots
\begin{equation}
\label{Prealzeros}
Q(\phi)=(\phi-u_1)(\phi-u_2)(\phi-u_3)(\phi-u_4) = \phi^4 - c \phi^2 + 2 b \phi + 2 d.
\end{equation}
Expanding (\ref{Prealzeros}) gives us the following relations between parameters $(b,c,d)$ and roots $(u_1,u_2,u_3,u_4)$ of $Q$: 
\begin{align}
\label{rG}
\left. \begin{array}{rl} 
u_1+u_2+u_3+u_4 &= 0, \\
u_1u_2+u_1u_3+u_1u_4+u_2u_3+u_2u_4+u_3u_4 &= -c, \\
u_1u_2u_3+u_1u_2u_4+u_1u_3u_4+u_2u_3u_4 &=-2b, \\
u_1u_2u_3u_4 &= 2d. \end{array} \right\} 
\end{align}
The following lemma identifies the periodic solutions 
of the system (\ref{first}), (\ref{second}), and (\ref{third}). 

\begin{lemma}
\label{lem-1}
There exist bounded periodic (nonconstant) solutions of the system (\ref{first}), (\ref{second}), and (\ref{third}) if and only if 
the four roots of $Q$ are real and distinct, which is equivalent to $(b,c) \in \Omega$ with $\Omega$ given by (\ref{domain}). By ordering the four roots of $Q$ as $u_4 < u_3 < u_2 < u_1$, the solution form is given by 
\begin{equation}
\label{form-1}
\phi(x) = u_4 + \dfrac{(u_2-u_4)(u_3-u_4)}{(u_2-u_4)-(u_2-u_3){\rm sn}^2(\nu x)},
\end{equation}
where 
\begin{equation}
\label{parameters-1}
\nu = \frac{1}{2} \sqrt{(u_1-u_3)(u_2-u_4)} \quad \mbox{\rm and} \quad 
k = \sqrt{\dfrac{(u_1-u_4)(u_2-u_3)}{(u_1-u_3)(u_2-u_4)}}.
\end{equation}
The fundamental period of the elliptic profile $\phi$ in $x$ is $2 \nu^{-1} K$ 
and its minimal and maximal values on $[-\nu^{-1}K,\nu^{-1}K]$ are attained at $\phi(0) = u_3$ and $\phi(\pm \nu^{-1} K) = u_2$, respectively. 
\end{lemma}

\begin{proof}
We consider separate cases of four, two, and no real roots of $Q$. 

\vspace{0.2cm}

{\em Four real simple roots of $Q$.} The first-order invariant 
$(\phi')^2 = Q(\phi)$ admits real solutions for $\phi$ if and only if $Q(\phi) \geq 0$, which is true in either of the following three intervals: $(-\infty,u_4]$, $[u_3,u_2]$, and $[u_1,\infty)$, provided that the roots of $Q$ have been ordered as $u_4 < u_3 < u_2 < u_1$. Solutions to (\ref{third}) are orbits on the phase plane $(\phi,\phi')$ given by the $(2d)$-level curve of the function 
$$
F(\phi,\phi') := (\phi')^2-\phi^4 + c\phi^2-2b\phi.
$$ 
Roots of $Q$ give the turning points, where the level curve of $F(\phi,\phi') = 2 d$ intersects with the $\phi$-line. Roots of $Q'$ give the equilibrium points 
$(\phi,0)$ on the phase plane $(\phi,\phi')$. It follows from the phase portrait on the $(\phi,\phi')$ plane, see Figure \ref{fig-phase}, that bounded periodic (nonconstant) solutions of the second-order equation (\ref{second}) exist if and only if the curve on the $(\phi,\phi')$ plane is located between two turning points and these turning points are not the equilibrium points. Hence, the intervals $(-\infty,u_4]$ and $[u_1,\infty)$ correspond to the unbounded solutions, whereas the bounded periodic (nonconstant) solutions exist in the only interval $[u_3,u_2]$. The four real roots of $Q$ exist for some levels of $F(\phi,\phi') = 2d$ if and only if 
$$
Q'(\phi) = 2 (2 \phi^3 - c \phi + b) 
$$
admits three real roots. If $c \leq 0$, there is only one real root of $Q'(\phi)$, whereas if $c > 0$, there exist three real roots of $Q'(\phi)$ if and only if $b \in (-b_c,b_c)$, where $b_c = \frac{\sqrt{2 c^3}}{\sqrt{27}}$ corresponds to the local maximum of the mapping $x \mapsto 2x^3 - c x$. This completes the proof of the existence region $\Omega$ for $(b,c)$ given by (\ref{domain}). To get the solution form (\ref{form-1}), we use 
formula 254.00 in \cite{BF1971} and obtain 
\begin{align*}
x &= \int_{u_3}^{\phi} \frac{dt}{\sqrt{(t-u_1)(t-u_2)(t-u_3)(t-u_4)}} \\
&= \dfrac{2}{\sqrt{(u_1-u_3)(u_2-u_4)}} F\left( \sin^{-1}\sqrt{\frac{(u_2-u_4)(\phi-u_3)}{(u_2-u_3)(\phi-u_4)}}, 
\sqrt{\dfrac{(u_2-u_3)(u_1-u_4)}{(u_1-u_3)(u_2-u_4)}} \right),
\end{align*}
where $F(\varphi,k)$ is the incomplete elliptic integral of the first kind:
\begin{equation*}
F(\varphi,k)=\int_0^{\varphi}\frac{d\alpha}{\sqrt{1-k^2\sin^2\alpha}}=\int_0^{\sin \varphi}\frac{dt}{(1-t^2)(1-k^2t^2)}.
\end{equation*}
Inverting $x = F(\varphi,k)$ by using Jacobi elliptic function 
${\rm sn}(x,k) = \sin(\varphi)$, we rewrite the explicit solution 
as (\ref{form-1}), where the elliptic modulus $k$ is dropped according to our convention. 

\begin{figure}[htb!]
	\centering
	\includegraphics[width=12cm,height=8cm]{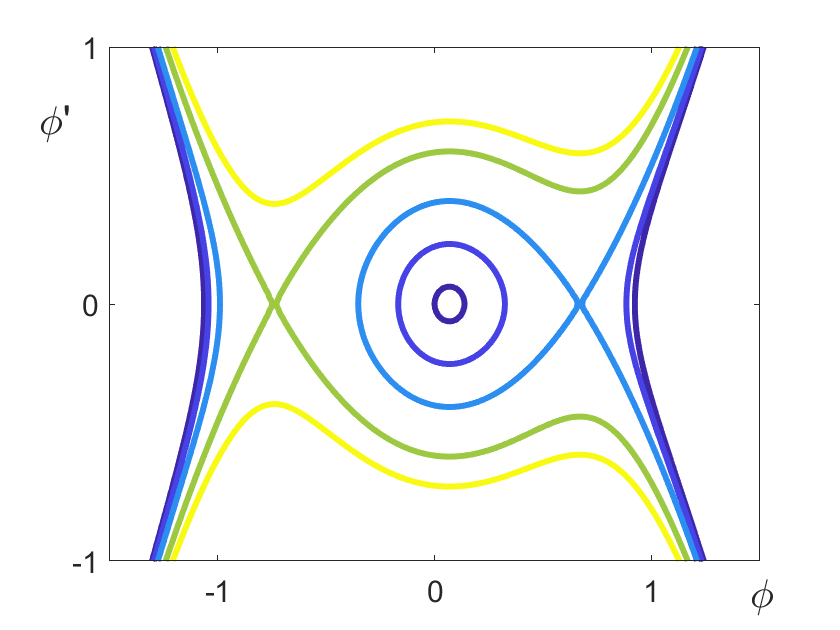}
	\caption{Phase portrait of the second-order equation (\ref{second}) from the level curves of $F(\phi,\phi')$ for $c = 1$ and $b = 0.25 b_c$.}
	\label{fig-phase}
\end{figure}

\vspace{0.2cm}

{\em Two real and two complex simple roots of $Q$.} Let us denote the real roots of $Q$ by $u_1$ and $u_2$ such that $u_2 < u_1$ and the complex-conjugate roots by $u_3$ and $u_4$. Since $Q(\phi) \sim \phi^4$ for large $\phi$ and the first-order invariant $(\phi')^2 = Q(\phi)$ admits real solutions for $\phi$ if and only if $Q(\phi) \geq 0$, the solutions only exist in the intervals $(-\infty,u_2]$ and $[u_1,\infty)$. Since each interval includes only one turning point, all solutions are unbounded in this case. 

\vspace{0.2cm}

{\em No real roots of $Q$.} We have $Q(\phi) > 0$ so that no turning points exist on $(-\infty,\infty)$. All solutions are unbounded in this case.

\vspace{0.2cm}

{\em Double real roots of $Q$.} In the case of four real roots (counting their multiplicities), if either $u_1 = u_2$ or $u_3 = u_4$, then $k = 1$ and the solution (\ref{form-1}) is non-periodic since $\lim\limits_{k\to 1} K(k) = \infty$ connecting the constant solutions $\phi(x) = u_1$ or $\phi(x) = u_4$ respectively, whereas if $u_2 = u_3$, then $k = 0$ and the solution (\ref{form-1}) is constant as $\phi(x) = u_2 = u_3$. In the case of two real roots and two complex roots (counting their multiplicities), if $u_1 = u_2$, then the only bounded solution is constant as $\phi(x) = u_1$. 

\vspace{0.2cm}

Thus, the bounded periodic (nonconstant) solution exists if and only if $Q$ admits four real simple roots.
\end{proof}

\subsection{Solution form parameterized by roots of the characteristic polynomial}

Let us consider the Lax system (\ref{LS}) for the traveling wave $u(x,t) = \phi(x+ct)$. Separating variables as $\varphi(x,t) = \psi(x+ct) e^{4 \mu t}$, we rewrite (\ref{LS}) in the form:
\begin{equation}
\label{spectral}
\frac{d}{dx} \psi = \left(\begin{array}{ll} i\zeta & \phi \\ \phi  & -i\zeta\end{array} \right) \psi
\end{equation}
and 
\begin{align}
\label{time}
4 \mu \psi = \left( \begin{array}{ll} 4i\zeta^3+2i\zeta \phi^2 - i c \zeta & 4\zeta^2 \phi - 2i\zeta \phi' +2\phi^3 - \phi'' - c \phi \\ 4\zeta^2 \phi + 2i \zeta \phi' + 2 \phi^3 - \phi'' - c \phi & -4i\zeta^3-2i\zeta \phi^2 + i c \zeta \end{array}\right) \psi,
\end{align}
where derivatives and the potentials depend on only one variable $x$ which stands for the traveling wave coordinate $x+ct$. Since (\ref{time}) is the algebraic system, admissible values of $\mu$ are found from the characteristic equation 
\begin{equation}
\label{char-eq}
\left| 
\begin{array}{ll} 4i\zeta^3+2i\zeta \phi^2 - ic \zeta - 4 \mu & 4\zeta^2 \phi - 2i\zeta \phi' + 2\phi^3 - \phi'' - c \phi \\ 4\zeta^2 \phi + 2i \zeta \phi' + 2 \phi^3 - \phi'' - c\phi & -4i\zeta^3-2i\zeta \phi^2 + i c \zeta - 4 \mu \end{array} 
\right| = 0.
\end{equation}
Expanding the determinant in (\ref{char-eq}) and using (\ref{second}) and (\ref{third}), we obtain $\mu^2 + P(\zeta) = 0$, where $P(\zeta)$ is the characteristic polynomial of the traveling waves in the form:
\begin{equation}
\label{char-poly}
P(\zeta) = \zeta^6 - \frac{c}{2} \zeta^4 + \frac{1}{16} (c^2 - 8d) \zeta^2 - \frac{b^2}{16}.
\end{equation}
It is customary to introduce roots $\{ \pm \zeta_1,\pm \zeta_2, \pm \zeta_3\}$ of the polynomial $P$ in the factorization 
\begin{equation}
\label{factor-P}
P(\zeta) = (\zeta^2 - \zeta_1^2) (\zeta^2 - \zeta_2^2) (\zeta^2 - \zeta_3^2),
\end{equation}
Expanding (\ref{factor-P}) and comparing with (\ref{char-poly}) yields the following relations between parameters $(b,c,d)$ and roots $(\zeta_1,\zeta_2,\zeta_3)$ of $P$:
\begin{align}
\label{rP}
\left. \begin{array}{rl} 
\displaystyle \zeta_1^2 + \zeta_2^2 + \zeta_3^2 &= \displaystyle  \frac{c}{2}, \vspace{0.2cm} \\ 
\displaystyle \zeta_1^2 \zeta_2^2 + \zeta_1^2 \zeta_3^2 + 
\zeta_2^2 \zeta_3^2 &= \displaystyle \frac{1}{16} (c^2 - 8d), \vspace{0.2cm} \\
\displaystyle \zeta_1^2 \zeta_2^2 \zeta_3^2 &= \displaystyle  \frac{b^2}{16}. \end{array} \right\}.
\end{align}
The following lemma gives a relation between roots of $Q$ and roots of $P$ in 
(\ref{rG}) and (\ref{rP}). The relations are well established in the literature, see \cite{Kam1,Kam2,Kam3,Kam4} but we still give the proof for the sake of completeness.

\begin{lemma}
	\label{lem-2}
Let $(u_1,u_2,u_3,u_4)$ be roots of $Q$ ordered as $u_4 < u_3 < u_2 < u_1$ in Lemma \ref{lem-1}. Roots $(\zeta_1,\zeta_2,\zeta_3)$ of $P$ are real and 
can be expressed from $(u_1,u_2,u_3,u_4)$ as 
\begin{equation}
\label{relation-roots}
\zeta_1 = \frac{1}{2} (u_1 + u_2), \qquad \zeta_2 = \frac{1}{2} (u_1 + u_3), \qquad \zeta_3 = \frac{1}{2} (u_2 + u_3),
\end{equation}
which satisfies ordering $\zeta_3 < \zeta_2 < \zeta_1$.
\end{lemma}

\begin{proof}
We will show that the parameterization (\ref{rP}) is equivalent to the parameterization (\ref{rG}) by using (\ref{relation-roots}). This verifies the validity of (\ref{relation-roots}). Since $(u_1,u_2,u_3,u_4)$ are real by Lemma \ref{lem-1}, then $(\zeta_1,\zeta_2,\zeta_3)$ are real. If $(u_1,u_2,u_3,u_4)$  satisfies ordering $u_4 < u_3 < u_2 < u_1$, then $(\zeta_1,\zeta_2,\zeta_3)$ satisfies ordering $\zeta_3 < \zeta_2 < \zeta_1$.

\vspace{0.2cm}

\underline{For the first equation of system (\ref{rP}),} we obtain with 
(\ref{relation-roots}) and $u_1 + u_2 + u_3 + u_4 = 0$ that 
\begin{align*}
2c &= 4(\zeta_1^2 + \zeta_2^2 + \zeta_3^2) \\
&= (u_1+u_2)^2 + (u_1+u_3)^2 + (u_2 + u_3)^2 \\
&= -(u_1+u_2)(u_3+u_4) - (u_1+u_3)(u_2+u_4) - (u_2 + u_3) (u_1+u_4) \\
&= -2u_1 u_2 - 2 u_1 u_3 - 2 u_1 u_4 - 2 u_2 u_3 - 2 u_2 u_4 - 2 u_3 u_4,
\end{align*}
which recovers the second equation of system (\ref{rG}). 

\vspace{0.2cm}

\underline{For the third equation of system (\ref{rP}),} 
we extract the square root for $b = 4 \zeta_1 \zeta_2 \zeta_3$ 
without loss of generality. It follows from 
(\ref{relation-roots}) and $u_1 + u_2 + u_3 + u_4 = 0$ that
\begin{align*}
2b &= 8 \zeta_1 \zeta_2 \zeta_3 \\
&= (u_1 + u_2) (u_1 + u_3) (u_2 + u_3) \\
&= -u_1(u_2+u_4) (u_2 + u_3) - u_2 (u_1 + u_3) (u_1 + u_4) \\
&= - u_1 u_2 u_3 - u_1 u_2 u_4 - u_1 u_3 u_4 - u_2 u_3 u_4 - u_1 u_2 (u_1 + u_2 + u_3 + u_4) \\
&= - u_1 u_2 u_3 - u_1 u_2 u_4 - u_1 u_3 u_4 - u_2 u_3 u_4,
\end{align*}
which recovers the third equation of system (\ref{rG}). 

\vspace{0.2cm}

\underline{For the second equation of system (\ref{rP}),} we first perform some algebraic manipulations. It follows from (\ref{rP}) and (\ref{relation-roots}) that 
\begin{align*}
4 \zeta_1^2 - c &= 2(\zeta_1^2 - \zeta_2^2 - \zeta_3^2) \\
&= \frac{1}{2} [(u_1+u_2)^2 - (u_1+u_3)^2 - (u_2+u_3)^2] \\
&= \frac{1}{2} [-(u_1+u_2)(u_3+u_4) + (u_1 + u_3)(u_2+u_4) + (u_2+u_3)(u_1+u_4)] \\
&= u_1u_2+u_3u_4,
\end{align*}
and similarly, 
\begin{align*}
4 \zeta_2^2 - c &= u_1u_3+u_2u_4,\\
4 \zeta_3^2 - c &= u_1u_4+u_2u_3.
\end{align*}
By using these relations, we transform the second equation of system (\ref{rP}) to the form:
\begin{align*}
c^2 - 8d &= 16(\zeta_1^2 \zeta_2^2 + \zeta_1^2 \zeta_3^2 + 
\zeta_2^2 \zeta_3^2) \\
&= (c + u_1 u_2 + u_3 u_4) (c + u_1 u_3 + u_2 u_4) + (c + u_1 u_2 + u_3 u_4) (c + u_1 u_4 + u_2 u_3) \\
& \qquad + (c + u_1 u_3 + u_2 u_4) (c + u_1 u_4 + u_2 u_3)  \\ 
&= 3c^2 + 2c (u_1 u_2 + u_1 u_3 + u_1 u_4 + u_2 u_3 + u_2 u_4 + u_3 u_4) \\
& \qquad +(u_1 u_2 + u_3 u_4) (u_1 u_3 + u_2 u_4) + (u_1 u_2 + u_3 u_4) (u_1 u_4 + u_2 u_3) \\
& \qquad  + (u_1 u_3 + u_2 u_4) (u_1 u_4 + u_2 u_3) \\
&= c^2 + u_1 u_2 u_3 (u_1 + u_2 + u_3) + u_1 u_2 u_4 (u_1 + u_2 + u_4) \\
& \qquad + u_1 u_3 u_4 (u_1 + u_3 + u_4) + u_2 u_3 u_4 (u_2 + u_3 + u_4) \\
&= c^2-4u_1u_2u_3u_4,
\end{align*}
which yields $2d = u_1 u_2 u_3 u_4$, that is, the fourth equation of system (\ref{rG}).
\end{proof} 

By using Lemma \ref{lem-2}, we rewrite the solution form (\ref{form-1}) in a different representation parameterized by roots $(\zeta_1,\zeta_2,\zeta_3)$. This result is given by the following lemma.

\begin{lemma}
	\label{lem-3}
The mapping $(\zeta_1,\zeta_2,\zeta_3) \to (b,c,d)$ given by (\ref{rP}) with $b = 4 \zeta_1 \zeta_2 \zeta_3$ is a diffeomorphism if roots of $P$ satisfy $\zeta_3 < \zeta_2 < \zeta_1$. The bounded periodic (nonconstant) solutions of Lemma \ref{lem-1} can be written in the form 
\begin{equation}
\label{form-2}
\phi(x) = \dfrac{2 (\zeta_1 + \zeta_3)(\zeta_2 + \zeta_3)}{(\zeta_1 + \zeta_3) -(\zeta_1 - \zeta_2){\rm sn}^2(\nu x)} - \zeta_1 - \zeta_2 - \zeta_3,
\end{equation}
with
\begin{equation}
\label{parameters-nu-k}
\nu = \sqrt{\zeta_1^2 - \zeta_3^2},  \quad
k = \sqrt{\dfrac{\zeta_1^2 - \zeta_2^2}{\zeta_1^2 - \zeta_3^2}}.
\end{equation} 		
\end{lemma}

\begin{proof}
Inverting (\ref{relation-roots}) with the constraint $u_1 + u_2 + u_3 + u_4 = 0$, we obtain
\begin{align}
\label{relation-inverse}
\left\{  \begin{array}{l} 
u_1 = \zeta_1 + \zeta_2 - \zeta_3,  \\
u_2 = \zeta_1 - \zeta_2 + \zeta_3, \\
u_3 = -\zeta_1 + \zeta_2 + \zeta_3, \\
u_4 = -\zeta_1 - \zeta_2 - \zeta_3. \end{array} \right.
\end{align}
The solution form (\ref{form-1})--(\ref{parameters-1}) transforms to the 
solution form (\ref{form-2})--(\ref{parameters-nu-k}) by using  substitution (\ref{relation-inverse}) and direct computations. 

Next, we show that the mapping $(\zeta_1,\zeta_2,\zeta_3) \to (b,c,d)$ given by (\ref{rP}) with $b = 4 \zeta_1 \zeta_2 \zeta_3$ is invertible for $\zeta_3 < \zeta_2 < \zeta_1$. The Jacobian of the mapping $(\zeta_1,\zeta_2,\zeta_3) \to (b,c,c^2 - 8d)$ is given by 
\begin{align*}
\frac{\partial (b,c,c^2-8d)}{\partial (\zeta_1,\zeta_2,\zeta_3)} &= 2^9 \left| 
\begin{array}{ccc} \zeta_2 \zeta_3 & \zeta_1 \zeta_3 & \zeta_1 \zeta_2 \\ 
\zeta_1 & \zeta_2 & \zeta_3 \\
\zeta_1 (\zeta_2^2 + \zeta_3^2) & \zeta_2 (\zeta_1^2 + \zeta_3^2) & 
\zeta_3 (\zeta_1^2 + \zeta_2^2) 
\end{array} \right|, \\
&= 2^9 \left[ \zeta_2^2 \zeta_3^2 (\zeta_2^2 - \zeta_3^2) - \zeta_1^2 \zeta_3^2 (\zeta_1^2 - \zeta_3^2) + \zeta_1^2 \zeta_2^2 (\zeta_1^2-\zeta_2^2) \right] \\
&= 2^9 (\zeta_1^2 - \zeta_2^2) (\zeta_1^2 - \zeta_3^2) (\zeta_2^2 - \zeta_3^2) \neq 0.
\end{align*}
Since the mapping $(\zeta_1,\zeta_2,\zeta_3) \to (b,c,c^2 - 8d)$ is a diffeomorphism, so is the mapping $(\zeta_1,\zeta_2,\zeta_3) \to (b,c,d)$.
\end{proof}

\begin{example}
	\label{ex-best-choice}
	In view of the scaling transformation (\ref{scal-transf}), it is tempting to fix $\nu = 1$ in the solution form (\ref{form-2}). However, if $\nu$ is fixed and $\zeta_1 = \sqrt{\zeta_3^2 + \nu^2}$, $\zeta_2 = \sqrt{\zeta_3^2 + (1-k^2) \nu^2}$ are uniquely defined by $(\zeta_3^2,k^2)$, we obtain
			\begin{equation}
	\label{par-b-c}
	\left\{  \begin{array}{l}
	c = 2(3 \zeta_3^2 + \nu^2 (2-k^2)), \\
	b^2 = 16 \zeta_3^2 (\zeta_3^2 + \nu^2) (\zeta_3^2 + \nu^2(1-k^2)).
	\end{array} \right.
	\end{equation}
The Jacobian of the mapping $(\zeta_3^2,k^2) \to (c,b^2)$ is computed from (\ref{par-b-c}) as 
	\begin{align*}
	\frac{\partial (c,b^2)}{\partial (\zeta_3^2,k^2)} &= 32 \left| 
	\begin{array}{cc} 3 & -\nu^2 \\ 
	3 \zeta_3^4 + 2 \nu^2 (2-k^2) \zeta_3^2 + \nu^4 (1-k^2) & -\nu^2 \zeta_3^2 (\zeta_3^2 + \nu^2) 
	\end{array} \right|, \\
	&= 32 \nu^4  (\zeta_3^2 (1-2k^2) + \nu^2 (1-k^2)).
	\end{align*}
The mapping $(\zeta_3^2,k^2) \to (c,b^2)$ for fixed $\nu > 0$ is only invertible for every $k^2 \in [0,\frac{1}{2}]$, which does not cover the entire existence region $\Omega$ in (\ref{domain}). This suggests that the solution form (\ref{form-2}) is not convenient for the reduction to the two-parameter form by fixing $\nu = 1$ 
according to the scaling transformation (\ref{scal-transf}).
\end{example}

\begin{example}
	\label{ex-particular-sol}
	The periodic solution (\ref{sn_potential}) with $b = 0$ is obtained for the choice 
\begin{equation}
\label{choice-snoidal}
	\zeta_1 = \frac{1}{2} (1+\tilde{k}), \quad \zeta_2 = \frac{1}{2} (1-\tilde{k}), \quad \zeta_3 = 0,
\end{equation}
	where $\tilde{k} \in (0,1)$ is the elliptic modulus of the solution $\phi(x) = \tilde{k} \sn(x,\tilde{k})$. It is interesting that the general periodic solution (\ref{form-2}) with the elliptic modulus $k \in (0,1)$ does not reduce to the solution form (\ref{sn_potential}) directly but recovers it after Landen's transformation. Indeed, it follows from (\ref{parameters-nu-k}) with (\ref{choice-snoidal}) that 
\begin{equation}
\label{choice-modulus}
	\nu = \frac{1}{2} (1+\tilde{k}), \qquad 
	\tilde{k} = \frac{2 \sqrt{\tilde{k}}}{1 + \tilde{k}}.
\end{equation}
	The solution form (\ref{form-2}) implies with the help of (\ref{choice-snoidal}) and (\ref{choice-modulus}) that 
	\begin{align*}
	\phi(x) &= \frac{1 - \tilde{k}^2}{1 + \tilde{k} - 2 \tilde{k} \sn^2(\nu x,k)} - 1\\
	&= \frac{1 - \tilde{k}^2}{1 - \tilde{k}^2 \sn^2(x,\tilde{k}) + \tilde{k} \cn(x,\tilde{k}) \dn(x,\tilde{k})} - 1,
	\end{align*}
where we have used 
$$
\sn^2(\nu x, k) = \frac{1 - \cn(2 \nu x, k)}{1 + \dn(2 \nu x,k)} = \frac{1}{2} \left[ 1 + \tilde{k} \sn^2(x,\tilde{k}) - \cn(x,\tilde{k}) \dn(x,\tilde{k}) \right]
$$
with the last equality due to the Lander transformation
$$
\cn\left( (1+\tilde{k}) x, \frac{2\sqrt{\tilde{k}}}{1 + \tilde{k}} \right) = \frac{\cn(x,\tilde{k}) \dn(x,\tilde{k})}{1 + \tilde{k} \sn^2(x,\tilde{k})}, \quad 
\dn\left( (1+\tilde{k}) x, \frac{2\sqrt{\tilde{k}}}{1 + \tilde{k}} \right) = \frac{1- \tilde{k} \sn^2(x,\tilde{k})}{1 + \tilde{k} \sn^2(x,\tilde{k})}.
$$
By using the fundamental relations (\ref{fund-elliptic}), we finally obtain 
\begin{align*}
\phi(x) &= \frac{1 - \tilde{k}^2}{\dn(x,\tilde{k}) (\dn(x,\tilde{k}) + 
	\tilde{k} \cn(x,\tilde{k})} - 1 \\
&= \frac{\dn(x,\tilde{k}) - 
\tilde{k} \cn(x,\tilde{k})}{\dn(x,\tilde{k})} - 1 \\
&= - \frac{\tilde{k} \cn(x,\tilde{k})}{\dn(x,\tilde{k})} = - \tilde{k} \sn(x + K(\tilde{k}),\tilde{k}),
\end{align*}
where we have used 
\begin{equation*}
\sn(x + K) = \frac{\cn(x)}{\dn(x)}.
\end{equation*}
This recovers the solution form $\phi(x) = \tilde{k} \sn(x,\tilde{k})$ up to the spatial translation and the sign flip.
\end{example}

\subsection{Solution form expressed by Weierstrass' elliptic functions}

The Weierstrass' elliptic function $\wp(x)$ with two fundamental periods $2\omega$ and $2\omega'$ is defined by the differential equation
\begin{equation}\label{ed}
(\wp'(x))^2=4\wp(x)^3-g_2\wp(x)-g_3 = 4(\wp(x) - e_1) (\wp(x)-e_2) (\wp(x)-e_3),
\end{equation}
where the roots 
\begin{align*}
e_1 = \wp(\omega), \qquad 
e_2 = \wp(\omega+\omega'), \qquad 
e_3 = \wp(\omega')
\end{align*}
are sorted as $e_3 < e_2 < e_1$. Weierstrass' and Jacobi's elliptic functions are related by
\begin{equation}
\label{rel-Jac-Wei}
\wp(x) = e_3 + \frac{e_1-e_3}{\sn^2(\sqrt{e_1-e_3} x,k)}, \qquad k = \sqrt{\frac{e_2-e_3}{e_1-e_3}},
\end{equation}
which yields the following relations betwen the fundamental periods 
\begin{equation}
\label{periods}
\omega = \frac{K(k)}{\sqrt{e_1-e_3}} \quad \mbox{\rm and} \quad 
\omega' = \frac{iK'(k)}{\sqrt{e_1-e_3}}.
\end{equation}
It follows from (\ref{ed}) that 
	\begin{align}
	\label{rG2}
	\left. \begin{array}{r} 
	\displaystyle	e_1+e_2+e_3 = 0, \\
	\displaystyle	e_1e_2+e_2e_3+e_1e_3 = -\frac{1}{4}g_2 \\
	\displaystyle	e_1e_2e_3 = \frac{1}{4}g_3. \end{array} \right\}
	\end{align}
The following lemma gives the solution form for $\phi(x)$ based on $\wp(x)$.  Although this representation is old, see \cite[p.13]{A1990} and \cite[p.453]{W}, it is useful for many computations in Sections \ref{sec-4} and \ref{sec-5}.

\begin{lemma}
	\label{lem-4} 
	The bounded periodic (nonconstant) solutions of Lemma \ref{lem-3} can be written as the linear fractional transformation of Weierstrass' elliptic function:
	\begin{equation}
	\label{lin-frac-trans}
	\phi(x) = \frac{\alpha_1 \wp(x) + \beta_1}{\wp(x) + \gamma_1},
	\end{equation}
	where 
	\begin{align*}
	\alpha_1 &:= \zeta_2 + \zeta_3 - \zeta_1,\\
	\beta_1 &:= \frac{1}{3} \zeta_1 (\zeta_1^2 - 2 \zeta_2^2 - 2 \zeta_3^2 - 3 \zeta_2 \zeta_3) + \frac{1}{3} (\zeta_2 + \zeta_3) (2 \zeta_1^2 - \zeta_2^2 - \zeta_3^2 + 3 \zeta_2 \zeta_3), \\
	\gamma_1 &:= \zeta_1 (\zeta_2 + \zeta_3) - \zeta_2 \zeta_3 -\frac{1}{3} (\zeta_1^2 + \zeta_2^2 + \zeta_3^2).
	\end{align*}
There exists $v \in \mathbb{C}$ in the rectangle $[-\omega,\omega] \times [-\omega',\omega']$  such that   
	\begin{equation}\label{pcW}
	\frac{c}{6}=\wp(v), \quad \frac{b}{2}=\wp'(v),
	\end{equation}
which allows us to rewrite the solution (\ref{lin-frac-trans}) in another form:
	\begin{equation}
\label{phinova}
\phi(x) = \frac{1}{2}\frac{\wp'(x-\frac{v}{2})+\wp'(x+\frac{v}{2})}{\wp(x-\frac{v}{2})-\wp(x + \frac{v}{2})}
\end{equation}
and
	\begin{equation}
	\label{philinha}
	\phi' (x) = \wp\left(x-\frac{v}{2} \right) - \wp\left(x+\frac{v}{2}\right).
	\end{equation}
\end{lemma}	
	
\begin{proof}
First, we show that parameters $(b,c,d)$ of the quartic polynomial $Q$ in (\ref{third}) are related to parameters $(g_2,g_3)$ of the cubic polynomial in (\ref{ed}) with the correspondence:
\begin{equation}
\label{rel-par-Jac-Wei}
	\left\{ \begin{array}{l} 
	\displaystyle g_2 = \frac{c^2}{12} + 2 d, \\ 
		\displaystyle g_3 = \frac{c^3}{216} - \frac{cd}{3} -\frac{b^2}{4}.
 \end{array} \right.
\end{equation}
To do so, we use the relations:
\begin{equation}
\label{rel-par-e-zeta}
	\left\{ \begin{array}{l} 
\displaystyle e_1 = \frac{1}{3} (\zeta_1^2 + \zeta_2^2 - 2 \zeta_3^2), \\
\displaystyle e_2 = \frac{1}{3} (\zeta_1^2 - 2 \zeta_2^2 + \zeta_3^2), \\
\displaystyle e_3 = \frac{1}{3} (-2 \zeta_1^2 + \zeta_2^2 + \zeta_3^2),
 \end{array} \right.
\quad \Rightarrow \quad 
\left\{ \begin{array}{l} 
e_1 - e_2 = \zeta_2^2 - \zeta_3^2, \\
e_1 - e_3 = \zeta_1^2 - \zeta_3^2, \\
e_2 - e_3 = \zeta_1^2 - \zeta_2^2.
\end{array} \right.
\end{equation}
By substituting (\ref{rel-par-e-zeta}) into (\ref{rG2}) and using (\ref{rP}), 
we obtain the relations (\ref{rel-par-Jac-Wei}) as follows:
\begin{align*}
g_2 &= -4(e_1 e_2 + e_2 e_3 + e_1 e_3) \\ 
&= -4\left[\left(\frac{c}{6}-\zeta_3^2\right)\left(\frac{c}{6}-\zeta_2^2\right)+\left(\frac{c}{6}-\zeta_3^2\right)\left(\frac{c}{6}-\zeta_1^2\right)+\left(\frac{c}{6}-\zeta_2^2\right)\left(\frac{c}{6}-\zeta_1^2\right)\right] \\
&= -4\left[ \frac{c^2}{12}-\frac{c}{3}(\zeta_1^2+\zeta_2^2+\zeta_3^2)+ (\zeta_1^2\zeta_2^2+\zeta_1^2\zeta_3^2+\zeta_2^2\zeta_3^2)\right] \\
&= -4\left[\frac{c^2}{12}-\frac{c^2}{6}+\frac{c^2-8d}{16}\right]\\
&=\frac{c^2}{12}+2d
\end{align*}
and
\begin{align*}
g_3 &= 4 e_1 e_2 e_3 \\
&= 4\left( \frac{c}{6}-\zeta_1^2\right)\left( \frac{c}{6}-\zeta_2^2\right)\left(\frac{c}{6}-\zeta_3^2 \right)\\ 
&= 4\left( \frac{c^3}{216}-\frac{c^2}{36}(\zeta_1^2+\zeta_2^2+\zeta_3^2)+\frac{c}{6}(\zeta_1^2\zeta_2^2+\zeta_1^2\zeta_3^2+\zeta_2^2\zeta_3^2)-\zeta_1^2\zeta_2^2\zeta_3^2\right)\\ 
&= 4\left( \frac{c^3}{216}-\frac{c^3}{72}+\frac{c^3}{96}-\frac{cd}{12}-\frac{b^2}{16} \right) \\
&= \frac{c^3}{216}-\frac{cd}{3}-\frac{b^2}{4}.
\end{align*} 
Comparing (\ref{form-2}) with (\ref{rel-Jac-Wei}) and (\ref{rel-par-e-zeta}) yields the linear fractional transformation (\ref{lin-frac-trans}). 
In order to derive (\ref{phinova}) and (\ref{philinha}), we obtain from (\ref{rel-par-Jac-Wei}) that
\begin{align*}
\left(\frac{b}{2}\right)^2 = 4\left(\frac{c}{6}\right)^3 - g_2\left(\frac{c}{6}\right) - g_3.
\end{align*}
Hence, $\left(\frac{c}{6},\frac{b}{2} \right)$ is a point on the elliptic curve for $\wp(x)$ so that there exists $v \in \mathbb{C}$ in the rectangle $[-\omega,\omega] \times [-\omega',\omega']$ which parameterizes $(b,c)$ in (\ref{pcW}). Consequently, equation (\ref{third}) can be rewritten in the following form
\begin{equation*}
(\phi')^2=[\phi^2-3\wp(v)]^2+4{\wp}'(v)\phi+g_2-12[\wp(v)]^2.
\end{equation*}
which is satisfied by (\ref{phinova}) and (\ref{philinha}) due to computations on pp.103-104 in \cite{A1990}.
\end{proof}

\begin{remark}
By using formula 8.177.1 in \cite{GR}, we can rewrite (\ref{phinova}) in the equivalent form:
\begin{equation}
\label{form-3}
\phi(x) = \zeta\left(x+\frac{v}{2}\right)-\zeta\left(x-\frac{v}{2}\right)-\zeta(v),
\end{equation}
where $\zeta$ is the Weierstrass' zeta function. Furthermore, by using formula 8.166.2 in \cite{GR}, we integrate the product of $\phi$ in (\ref{phinova}) and $\phi'$ in (\ref{philinha}) to obtain 	
\begin{equation}
\label{form-3-squared}
\phi^2(x) = \wp\left(x + \frac{v}{2}\right) + \wp\left(x - \frac{v}{2}\right) + \wp(v).
\end{equation}
\end{remark}

\subsection{Solution form expressed by Jacobi's theta functions}

For the soluton form (\ref{form-2}) in Lemma \ref{lem-3}, it is convenient to use the variable $z := \nu x$. Since $\phi$ is periodic in $z = \nu x$ with the fundamental periods $2 K$ and $2i K'$, we are looking for poles and zeros in the fundamental rectangle $[-K,K] \times [-iK',iK']$. This rectangle in variable $z$ is equivalent to the rectangle $[-\omega,\omega] \times [-\omega',\omega']$ for the Weierstrass' elliptic function $\wp$ in variable $x$ due to (\ref{periods}) and (\ref{rel-par-e-zeta}). 

From the theory of elliptic functions, we adopt the following three propositions, where the numbers of zeros and poles are defined according to their multiplicity and subject to the periodic boundary conditions on the boundaries of the fundamental rectangles.

\begin{proposition}
	\label{prop-1}
	For every $c \in \mathbb{C}$, there exist only two solutions of the elliptic equation $\wp(x) = c$ in the rectangle $[-\omega,\omega] \times [-\omega',\omega']$ or equivalently, only two solutions of the elliptic equation $\sn^2(z) = c$ in the rectangle $[-K,K] \times [-iK',iK']$.
\end{proposition}

\begin{proof}
Since the Weierstrass function $\wp$ is an elliptic function of the second order (pages 8 and 12 in \cite{A1990}),  the assertion follows from 20-13 on page 432 in \cite{W}.
\end{proof}

\begin{proposition}
	\label{prop-2}
	For every elliptic function, the number of zeros and poles in the rectangle $[-K,K] \times [-iK',iK']$ coincide. 
\end{proposition}

\begin{proof}
	The assertion follows by 20-13 on page 432 in \cite{W}.
\end{proof}

\begin{proposition}
	\label{prop-3}
Let $f(x) : \mathbb{C} \to \mathbb{C}$ be an elliptic function with two fundamental periods $2\omega, 2\omega' \in \mathbb{C}$ and with $N$ zeros $\{ x'_1, x'_2, \dots, x'_N \}$ and $N$ poles $\{ x_1, x_2, \dots, x_N \}$ inside $[-\omega,\omega] \times [-\omega',\omega']$. There exists a constant $C \in \mathbb{C}$ such that 
\begin{equation}
\label{factorzation-f}
f(x) = C \frac{H(x - x_1') H(x-x_2') \dots H(x-x'_N)}{H(x - x_1) H(x-x_2) \dots H(x-x_N)},
\end{equation}
where $H(x)$ is Jacobi's theta function given by (\ref{Jacobi-theta}).
\end{proposition}

\begin{proof}
	See Sections 14 and 19 in \cite{A1990}.
\end{proof}

By Lemma \ref{lem-4}, the profile $\phi$ is an elliptic function, see  (\ref{lin-frac-trans}), (\ref{phinova}), and (\ref{form-3}). The following lemma specifies the number and locations of zeros and poles of $\phi$ in the fundamental rectangle. To avoid confusion, we write explicitly the 
elliptic modulus $k \in (0,1)$ and use $k' = \sqrt{1-k^2}$ in Jacobi's elliptic functions.

\begin{lemma}
	\label{lem-5}
Assume that $0 < \zeta_3 < \zeta_2 < \zeta_1$. There exist exactly two simple poles of $\phi(x)$ for $z = \nu x$ in $[-K,K] \times [-iK',iK']$ at $\pm (iK' + \alpha)$ with $\alpha \in (0,K)$ given by 
\begin{equation}
\label{root-1}
\sn(\alpha,k) = \sqrt{\frac{\zeta_1 - \zeta_3}{\zeta_1 + \zeta_2}}
\end{equation}
and two simple zeros of $\phi$ in $[-K,K] \times [-iK',iK']$ at $\pm \beta$, 
where $\beta \in (0,K)$ is real for $\zeta_1 > \zeta_2 + \zeta_3$ 
given by 
\begin{equation}
\label{root-2}
\sn(\beta,k) = \sqrt{\frac{(\zeta_1 + \zeta_3)(\zeta_1 - \zeta_2 - \zeta_3)}{
	(\zeta_1 - \zeta_2) (\zeta_1 + \zeta_2 + \zeta_3)}}, 
\end{equation}
$\beta \in i (0,K')$ is purely imaginary for $\zeta_1 < \zeta_2 + \zeta_3$ given by 
	\begin{equation}
\label{root-3}
\sn(-i\beta,k') = \sqrt{\frac{(\zeta_1 + \zeta_3)(\zeta_2 + \zeta_3 - \zeta_1)}{
	(\zeta_2 + \zeta_3) (\zeta_1 - \zeta_2 + \zeta_3)}},
\end{equation}
and $\beta = 0$ for $\zeta_1 = \zeta_2 + \zeta_3$ (in which case, there is only one double zero). 
\end{lemma}

\begin{proof}
It follows from (\ref{form-2}) that the poles of $\phi$ are obtained from roots of the following equation:
\begin{equation}
\label{eq-1}
\sn^2(z,k) = \frac{\zeta_1 + \zeta_3}{\zeta_1 - \zeta_2} > 1.
\end{equation}
By using the formula 
$$
\sn(z+iK',k) = \frac{1}{k \sn(z,k)},
$$
we can rewrite (\ref{eq-1}) for $z = iK' + \alpha$:
\begin{equation*}
\sn^2(\alpha,k) = \frac{1}{k^2} \frac{\zeta_1-\zeta_2}{\zeta_1 + \zeta_3} = \frac{\zeta_1 - \zeta_3}{\zeta_1 + \zeta_2}  \in (0,1),
\end{equation*}
By Proposition \ref{prop-1}, there exist only two solutions of this equation in $[-K,K] \times [-iK',iK']$ and since $\sn^2(\alpha,k) \in (0,1)$, the roots are located 
symmetrically at $\pm \alpha$, where $\alpha \in (0,K)$ is uniquely obtained from (\ref{root-1}). Due to $2iK'$ periodicity of $\sn^2(z,k)$, the two roots can be equivalently placed  at $\pm (iK' + \alpha)$.

It follows from (\ref{form-2}) that the zeros of $\phi$ are obtained from roots of the following equation:
\begin{equation}
\label{eq-2}
\sn^2(z,k) = \frac{(\zeta_1 + \zeta_3)(\zeta_1 - \zeta_2 - \zeta_3)}{
	(\zeta_1 - \zeta_2) (\zeta_1 + \zeta_2 + \zeta_3)}.
\end{equation}
If $\zeta_1 = \zeta_2 + \zeta_3$, then the only solution of (\ref{eq-2}) is a double zero at $0$. If $\zeta \neq \zeta_2 + \zeta_3$, two solutions of (\ref{eq-2}) correspond to two simple zeros.
\begin{itemize}
	\item If $\zeta_1 > \zeta_2 + \zeta_3$, then
	\begin{equation*}
\frac{(\zeta_1 + \zeta_3)(\zeta_1 - \zeta_2 - \zeta_3)}{
		(\zeta_1 - \zeta_2) (\zeta_1 + \zeta_2 + \zeta_3)} \in (0,1)
	\end{equation*}
	since  
	$$
	(\zeta_1 + \zeta_3)(\zeta_1 - \zeta_2 - \zeta_3) <
	(\zeta_1 - \zeta_2) (\zeta_1 + \zeta_2 + \zeta_3)
	$$
	is satisfied due to 
	$$
	(\zeta_2 + \zeta_3) (\zeta_1 - \zeta_2 + \zeta_3) > 0.
	$$
As a result, the roots of (\ref{eq-2}) are real and located symmetrically at $\pm \beta$, where $\beta \in (0,K)$ is uniquely obtained from (\ref{root-2}).
	
	\item If $\zeta_1 < \zeta_2 + \zeta_3$, then 
	\begin{equation*}
\frac{(\zeta_1 + \zeta_3)(\zeta_1 - \zeta_2 - \zeta_3)}{
	(\zeta_1 - \zeta_2) (\zeta_1 + \zeta_2 + \zeta_3)} < 0.
\end{equation*}	
By using the formula
	$$
	\sn(iz,k) = \frac{i \sn(z,k')}{\cn(z,k')},
	$$
we rewrite (\ref{eq-2}) in the equivalent form for $z = i z'$:
	\begin{equation*}
	\sn^2(z',k') = \frac{(\zeta_1 + \zeta_3)(\zeta_2 + \zeta_3 - \zeta_1)}{
		(\zeta_2 + \zeta_3) (\zeta_1 - \zeta_2 + \zeta_3)} \in (0,1),
	\end{equation*}
where the right-hand side belongs to $(0,1)$ since 
	$$
	(\zeta_1 + \zeta_3)(\zeta_2 + \zeta_3 - \zeta_1) < 
	(\zeta_2 + \zeta_3) (\zeta_1 - \zeta_2 + \zeta_3)
	$$
	is satisfied due to 
	$$
	(\zeta_1 - \zeta_2) (\zeta_1 + \zeta_2 + \zeta_3) > 0.
	$$
	As a result, the roots of (\ref{eq-2}) are purely imaginary and located symmetrically at $\pm \beta$, where $\beta \in i (0,K')$ is uniquely obtained from (\ref{root-3}).
\end{itemize}
These computations complete the proof of the assertions. 
\end{proof}

Figure \ref{fig-zeros} shows the location of poles and zeros of $\phi(x)$ 
for $z = \nu x$ in $[-K,K] \times [-iK',iK']$.

\begin{figure}[htb!]
	\includegraphics[width=7.5cm,height=6cm]{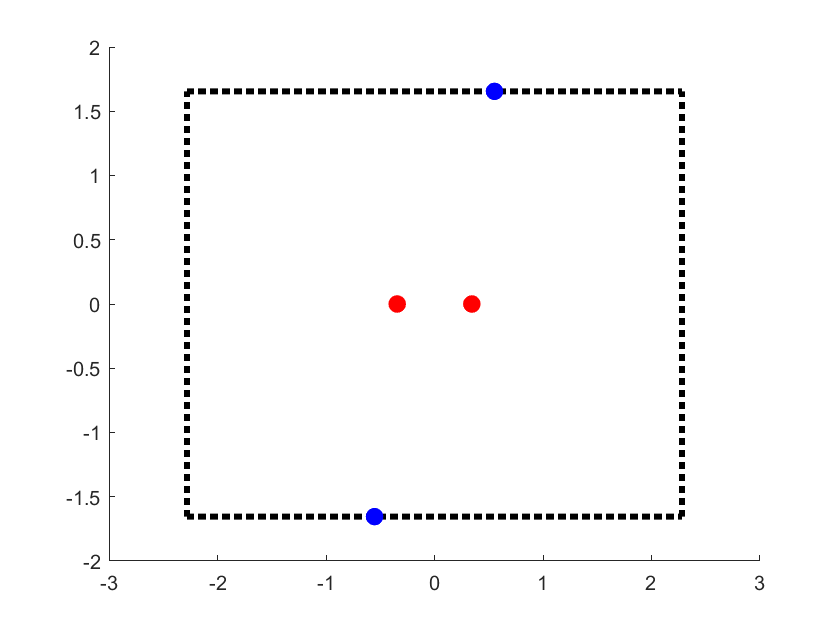}
	\includegraphics[width=7.5cm,height=6cm]{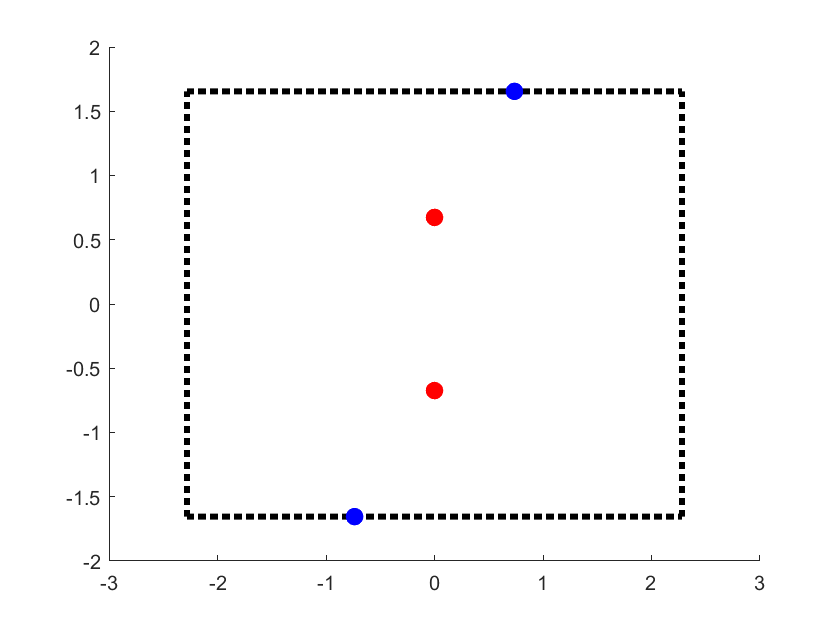}
	\caption{The rectangle $[-K,K] \times [-iK',iK']$ in the $z$-plane (dotted black) with zeros of $\phi$ (red dots) and poles of $\phi$ (blue dots). The zeros of $\phi$ are located on the real axis for $\zeta_1 > \zeta_2 + \zeta_3$ (left) and on the imaginary axis for $\zeta_1 < \zeta_2 + \zeta_3$ (right) whereas the poles of $\phi$ are located at $\pm (i K' + \alpha)$ with $\alpha \in (0,K)$ given by (\ref{root-1}).}
	\label{fig-zeros}
\end{figure}

\begin{lemma}
	\label{lem-v-alpha-connection}
The point $v \in [-\omega,\omega] \times [-\omega',\omega']$ in Lemma \ref{lem-4} is related to the value of $\alpha$ in Lemma \ref{lem-5} by the following correspondence:
\begin{equation}
\label{v-alpha-correspondence}
\frac{v}{2} = -\frac{iK' + \alpha}{\sqrt{e_1-e_3}}.
\end{equation}
\end{lemma}

\begin{proof}
Since $\wp(x) = x^{-2} + \mathcal{O}(1)$ as $x \to 0$, it follows from (\ref{philinha}) that $x = \pm \frac{v}{2}$ are simple poles of $\phi(x)$ such that 
\begin{align}
\label{singularity-1}
\lim_{x \to \pm \frac{v}{2}} \left( x \mp \frac{v}{2} \right) \phi(x) = \mp 1.
\end{align}
Expanding (\ref{form-2}) in $z$ near $z = i K' + \alpha$ yields
\begin{align}
\lim_{\nu x \to i K' + \alpha} ( \nu x - i K' - \alpha) \phi(x) 
&= -\frac{(\zeta_1 + \zeta_3) (\zeta_2 + \zeta_3)}{(\zeta_1 - \zeta_2) \sn(iK'+\alpha) \cn(iK' + \alpha) \dn(iK' + \alpha)} \notag \\
&= \frac{(\zeta_1 + \zeta_3) (\zeta_2 + \zeta_3) k^2 \sn^3(\alpha)}{(\zeta_1 - \zeta_2) \cn(\alpha) \dn(\alpha)},
\label{singularity-3}
\end{align}
where we have used translation formulas for Jacobi's elliptic functions 
\begin{equation}
\label{shift-by-iK}
\sn(z+iK') = \frac{1}{k \sn(z)}, \quad 
\cn(z+iK') = \frac{-i \dn(z)}{k \sn(z)}, \quad 
\dn(z+iK') = \frac{-i \cn(z)}{\sn(z)}.
\end{equation}
The limit (\ref{singularity-3}) is computed by using  (\ref{parameters-nu-k-again}) and  (\ref{elliptic-alpha}),
\begin{align}
\label{singularity-2}
\lim_{\nu x \to i K' + \alpha} ( \nu x - i K' - \alpha) \phi(x) = \nu.
\end{align}
Since $\nu = \sqrt{e_1-e_3}$, comparing (\ref{singularity-1}) and (\ref{singularity-2}) implies (\ref{v-alpha-correspondence}). 
\end{proof}

Next, we incorporate the poles and zeros of the elliptic function $\phi$ in Lemma \ref{lem-5} and obtain the solution form (\ref{gen-theta}) in Theorem \ref{theorem-wave}. This is done by using Propositions \ref{prop-2} and \ref{prop-3} with the factorization formula (\ref{factorzation-f}) for $N = 2$.

\begin{lemma}
	\label{lem-6}
Assume $0 < \zeta_3 < \zeta_2 < \zeta_1$ and $\zeta_1 \neq \zeta_2 + \zeta_3$.
Let $\alpha \in (0,K)$ and $\beta \in (0,K) \cup i (0,K')$ be defined as in Lemma \ref{lem-5}. Then, the periodic solution with the profile $\phi$ in Lemma \ref{lem-4} is given by 
\begin{equation}
\label{gen-theta-rewrite}
\phi(x) = (\zeta_1 - \zeta_2 - \zeta_3)  \frac{H(\nu x- \beta) H(\nu x+ \beta) \Theta^2(\alpha)}{ \Theta(\nu x- \alpha) \Theta(\nu x + \alpha) H^2(\beta)}.
\end{equation}
\end{lemma}

\begin{proof}
By Proposition \ref{prop-3} with $N = 2$, the elliptic function $\phi$ is 
factorized as a quotient of the products of Jacobi's theta function $H$ centered at two zeros and two poles given by Lemma \ref{lem-5}:
\begin{equation*}
\phi(x) = C \frac{H(\nu x- \beta) H(\nu x+ \beta)}{ H(\nu x- iK' -\alpha) H(\nu x + iK' + \alpha)},
\end{equation*}
where $C$ is a specific constant to be determined uniquely. Since 
\begin{equation}
\label{H-Theta-connection}
H(z + i K') = i e^{\frac{\pi K'}{4 K}} e^{-\frac{i \pi z}{2K}} \Theta(z),
\end{equation}
we obtain 
$$
H(\nu x- iK' - \alpha) H(\nu x + iK' + \alpha) = e^{\frac{\pi K'}{2 K}} 
e^{-\frac{i \pi \alpha}{K}} \Theta(\nu x - \alpha) \Theta(\nu x + \alpha),
$$
which yields
\begin{equation*}
\phi(x) = \tilde{C} \frac{H(\nu x- \beta) H(\nu x+ \beta)}{ \Theta(\nu x- \alpha) \Theta(\nu x + \alpha)},
\end{equation*}
with $\tilde{C} = C e^{-\frac{\pi K'}{2K}} e^{\frac{i \pi \alpha}{K}}$. 
Since $\phi(0) = u_3 = \zeta_2 + \zeta_3 - \zeta_1$, we obtain 
the unique expression for $\tilde{C}$ by
\begin{equation*}
\tilde{C} = (\zeta_1 - \zeta_2 - \zeta_3) \frac{\Theta^2(\alpha)}{H^2(\beta)},
\end{equation*}
since $H$ is odd and $\Theta$ is even. This yields the analytical representation (\ref{gen-theta-rewrite}). 
\end{proof} 

\begin{remark}
For $\zeta_1 < \zeta_2 + \zeta_3$, we note that $\beta \in i (0,K')$ but 
$H(\nu x - \beta) H(\nu x + \beta)$ is real for real $x$. By using the squared relation between the Jacobi theta functions,
\begin{equation*}
H(x+y) H(x-y) \Theta^2(0) = H^2(x) \Theta^2(y) - \Theta^2(x) H^2(y),
\end{equation*}
we can rewrite the solution form (\ref{gen-theta-rewrite}) for $\zeta_1 < \zeta_2 + \zeta_3$ in the equivalent form:
\begin{equation*}
\phi(x) = (\zeta_1 - \zeta_2 - \zeta_3) \; \frac{\Theta^2(\alpha)}{H^2(\beta)} \; 
\frac{H^2(\nu x) \Theta^2(\beta) - \Theta^2(\nu x) H^2(\beta)}{\Theta^2(0) \Theta(\nu x- \alpha) \Theta(\nu x + \alpha)}.
\end{equation*} 
Using 
$$
H(\beta) = \sqrt{k} \sn(\beta) \Theta(\beta) \quad \mbox{\rm and} \quad 
\sn(i\gamma,k) = \frac{i \sn(\gamma,k')}{\cn(\gamma,k')},
$$ 
with $\beta = i \gamma$ and $\gamma \in (0,K')$, we can rewrite the solution form as 
\begin{equation}
\label{form-6}
\phi(x) = (\zeta_2 + \zeta_3 - \zeta_1) \frac{\Theta^2(\alpha) \Theta^2(\nu x)}{\Theta^2(0) \Theta(\nu x- \alpha) \Theta(\nu x + \alpha)} \frac{\sn^2(\nu x,k) \cn^2(\gamma,k') + \sn^2(\gamma,k')}{\sn^2(\gamma,k')},
\end{equation}
which only involve the real-valued elliptic functions. 
\end{remark}

\begin{remark}
	If $\zeta_1 = \zeta_2 + \zeta_3$, then $\beta = 0$ and the solution form (\ref{gen-theta-rewrite}) is undetermined. Since $H(\beta) = \sqrt{k} \sn(\beta) \Theta(\beta)$, we use 
	(\ref{root-2}) and obtain (\ref{wave-deg}) in the limit $\beta \to 0$.
	The same solution form (\ref{wave-deg}) follows from (\ref{root-3}) with $\beta = i \gamma$ and (\ref{form-6}) in the limit $\gamma \to 0$.
\end{remark}

\begin{remark}
It follows from (\ref{form-3}) by using formula 8.193.2 in \cite{GR} that 
\begin{align}
\phi(x) &=  \nu \left[\frac{H'(\nu (x+\frac{v}{2}))}{H(\nu(x+\frac{v}{2}))} -\frac{H'(\nu (x-\frac{v}{2}))}{H(\nu (x-\frac{v}{2}))}-\frac{H'(\nu v)}{H(\nu v)}\right] \notag \\
&= \nu \left[ \frac{\Theta'(\nu x - \alpha)}{\Theta(\nu x - \alpha)} 
- \frac{\Theta'(\nu x + \alpha)}{\Theta(\nu x + \alpha)} 
+ \frac{H'(2\alpha)}{H(2\alpha)} \right],
\label{form-4}
\end{align}
where we have used (\ref{v-alpha-correspondence}), 
(\ref{H-Theta-connection}), and 
$$
H(z + 2 i K') = -e^{-\frac{\pi K'}{K}} e^{-\frac{i \pi z}{K}} H(z).
$$
The solution form (\ref{form-4}) is useful in the proof of Theorem \ref{theorem-breather}.
\end{remark}

\begin{example}
	For $k = 1$, the periodic solution with the elliptic profile of Theorem \ref{theorem-wave} transforms to the soliton solution with the hyperbolic profile of Example \ref{ex-hyperbolic-theta}. In this case, $\zeta_2 = \zeta_3$ and solutions of (\ref{root-1}) and (\ref{root-2}) are found from 
	$$
	\cosh^2(\alpha) = \frac{\zeta_1 + \zeta_2}{2\zeta_2}, \quad \cosh^2(\beta) = \frac{(\zeta_1 - \zeta_2) (\zeta_1 + 2 \zeta_2)}{2 \zeta_1 \zeta_2}.
	$$
Both solution forms (\ref{form-2}) and (\ref{gen-theta-rewrite}) are equivalent and reduce to each other as in
	\begin{align*}
\phi(x) &= \frac{4 \zeta_2 (\zeta_1 + \zeta_2) \cosh^2(\nu x)}{\cosh^2(\nu x) (\zeta_1 + \zeta_2) - \sinh^2(\nu x) (\zeta_1 - \zeta_2)} - \zeta_1 - 2 \zeta_2 \\
&= \zeta_1 - \frac{2(\zeta_1^2 - \zeta_2^2)}{\cosh^2(\nu x) (\zeta_1 + \zeta_2) - \sinh^2(\nu x) (\zeta_1 - \zeta_2)} \\
&= \zeta_1 - \frac{2 \nu^2}{\zeta_1 + \zeta_2 \cosh(2 \nu x)},
\end{align*}
and
	\begin{align*}
	\phi(x) &= (\zeta_1 - 2 \zeta_2) \frac{\cosh^2(\alpha) \sinh(\nu x - \beta) \sinh(\nu x + \beta)}{\sinh^2(\beta) \cosh(\nu x - \alpha) \cosh(\nu x + \alpha)} \\
	&= \frac{\sinh^2(\nu x) (\zeta_1 - \zeta_2) (\zeta_1 + 2 \zeta_2) - \cosh^2(\nu x) (\zeta_1 + \zeta_2) (\zeta_1 - 2 \zeta_2)}{\cosh^2(\nu x) (\zeta_1 + \zeta_2) - \sinh^2(\nu x) (\zeta_1 - \zeta_2)} \\
	&= \frac{\zeta_1 \zeta_2 \cosh(2 \nu x) - \zeta_1^2 + 2 \zeta_2^2}{ \zeta_1 + \zeta_2 \cosh(2 \nu x)} \\
	&= \zeta_1 - \frac{2 \nu^2}{\zeta_1 + \zeta_2 \cosh(2 \nu x)},
	\end{align*}
	where $\nu = \sqrt{\zeta_1^2 - \zeta_2^2}$.
	\label{ex-hyperbolic}
\end{example}

The proof of Theorem \ref{theorem-wave} is accomplished with the results of 
Lemmas \ref{lem-1}, \ref{lem-3}, \ref{lem-5}, and \ref{lem-6}.

\section{Kink breather}
\label{sec-4}

Let $\psi \in C^1(\mathbb{R},\mathbb{C}^2)$ be a solution of the Lax system (\ref{spectral}) and (\ref{time}). It follows from (\ref{char-eq}) that  $\mu^2 + P(\zeta) = 0$ with $P(\zeta)$ being the characteristic polynomial 
in (\ref{char-poly}). As in (\ref{factor-P}), we factorize $P(\zeta)$ by $\{\pm \zeta_1, \pm \zeta_2, \pm \zeta_3\}$ satisfying $\zeta_3 < \zeta_2 < \zeta_1$. One solution $\psi = (p,q)^T$ of the Lax system (\ref{spectral}) and (\ref{time}) is defined by 
\begin{equation}
\label{mu}
\mu = i \sqrt{P(\zeta)} = i \sqrt{(\zeta^2 - \zeta_1^2) (\zeta^2 - \zeta_2^2) (\zeta^2 - \zeta_3^2)}.
\end{equation}
Another solution $\psi = (p^*,q^*)^T$ of the Lax system (\ref{spectral}) and (\ref{time}) is defined by $\mu = -i \sqrt{P(\zeta)}$. 

It follows from (\ref{time}) that the quotient $\rho = q/p$ for the solution $\psi = (p,q)^T$ satisfies  
\begin{equation}
\label{rho-basic}
\rho = -\frac{4 i \zeta^3 + i \zeta (2 \phi^2 - c) - 4\mu}{4 \zeta^2 \phi - 2 i \zeta \phi' - b} = \frac{4 \zeta^2 \phi + 2 i \zeta \phi' - b}{4 i \zeta^3 + i \zeta (2 \phi^2 - c) + 4\mu},
\end{equation}
where we have used the second-order equation (\ref{second}) for the profile $\phi$. By using  $b = 4 \zeta_1 \zeta_2 \zeta_3$, $c = 2 (\zeta_1^2 + \zeta_2^2 + \zeta_3^2)$, and $\mu = i \sqrt{P(\zeta)}$, we rewrite the two quotients in the equivalent form:
\begin{align}
\label{rho1}
\rho &= -i \frac{\zeta^3 + \frac{1}{2} \zeta (\phi^2 - \zeta_1^2 - \zeta_2^2 - \zeta_3^2) - \sqrt{P(\zeta)}}{\zeta^2 \phi - \frac{i}{2} \zeta \phi' - \zeta_1 \zeta_2 \zeta_3} \\
\label{rho2}
&= -i \frac{\zeta^2 \phi + \frac{i}{2} \zeta \phi' - \zeta_1 \zeta_2 \zeta_3}{\zeta^3 + \frac{1}{2} \zeta (\phi^2 - \zeta_1^2 - \zeta_2^2 - \zeta_3^2) + \sqrt{P(\zeta)}}.
\end{align}
If $\rho$ is defined by either (\ref{rho1}) or (\ref{rho2}), the first component of the eigenfunction $\phi = (p,q)^T$ can be found from the first-order  equation 
\begin{equation}
\label{p-eq}
\partial_x p = (i \zeta + \phi \rho) p,
\end{equation}
which follows from (\ref{spectral}). 
For $\zeta = 0$, the two solutions $\psi = (p,q)^T$ and $\psi = (p^*,q^*)^T$ are found explicitly by using the representation (\ref{form-4}) for the elliptic profile $\phi$. The following lemma  recovers expressions  (\ref{p-g-gen}) and (\ref{p-g-gen-second}) which are written in original variables $(x,t)$. 

\begin{lemma}
\label{lem-kink-1}
Assume $0 < \zeta_3 < \zeta_2 < \zeta_1$ and consider the periodic solution of Theorem \ref{theorem-wave} with the elliptic profile $\phi$. The two linearly independent solutions $\psi = (p_0,q_0)^T$ and $\psi = (p_0^*,q_0^*)^T$ of the 
Lax system (\ref{spectral}) and (\ref{time}) with $\zeta = 0$ are given by 
\begin{equation}
\label{eigenfunction-zero-2}
\mu = -\zeta_1 \zeta_2 \zeta_3 : \quad \left( \begin{array}{c} p_0 \\ q_0 \end{array} \right)  = e^{s_0 x} \frac{\Theta(\nu x - \alpha)}{\Theta(\nu x + \alpha)} \left( \begin{array}{c} 1 \\ 1 \end{array} \right)
\end{equation}
and 
\begin{equation}
\label{eigenfunction-zero-1}
\mu = \zeta_1 \zeta_2 \zeta_3 : \quad \left( \begin{array}{c} p_0^* \\ q_0^* \end{array} \right) = e^{-s_0x} \frac{\Theta(\nu x + \alpha)}{\Theta(\nu x - \alpha)}  \left( \begin{array}{c} 1 \\ -1 \end{array} \right), 
\end{equation}
where $s_0 = \frac{\nu H'(2\alpha)}{H(2 \alpha)}$.
\end{lemma}

\begin{proof}
	Since $\zeta = 0$, the two solutions of $\mu^2 + P(\zeta) = 0$ are 
	$\mu = \zeta_1 \zeta_2 \zeta_3$ and $\mu = -\zeta_1 \zeta_2 \zeta_3$.
If $\mu = -\zeta_1 \zeta_2 \zeta_3$, then $\rho = 1$ from (\ref{rho-basic}) 
since $b = 4 \zeta_1 \zeta_2 \zeta_3$. By using (\ref{p-eq}) with $\zeta = 0$ and $\rho = 1$ and computing the integral with (\ref{form-4}), we obtain 
	\begin{align*}
p_0(x) = e^{\int \phi dx} = e^{s_0 x} \frac{\Theta(\nu x - \alpha)}{\Theta(\nu x + \alpha)}, \quad s_0 = \frac{\nu H'(2\alpha)}{H(\alpha)},
\end{align*}	
which yields (\ref{eigenfunction-zero-2}). Similarly, if $\mu = \zeta_1 \zeta_2 \zeta_3$, then $\rho = -1$ from (\ref{rho-basic}) so that by using (\ref{p-eq}) with $\zeta = 0$ and $\rho = -1$ and (\ref{form-4}), we obtain 
	\begin{align*}
	p_0^*(x) = e^{-\int \phi dx} = e^{-s_0 x} \frac{\Theta(\nu x + \alpha)}{\Theta(\nu x - \alpha)}, 
	\end{align*}
which yields (\ref{eigenfunction-zero-1}). 
\end{proof}

In original variables $(x,t)$ of the mKdV equation (\ref{mkdv}), we should introduce two coordinates. The traveling periodic wave $u(x,t) = \phi(x+ct)$ has the coordinate $\xi = x + ct$, whereas the exponential factor of the eigenfunctions $\varphi(x,t) = \psi(x+ct) e^{4 \mu t}$ has 
the coordinate 
\begin{equation*}
\eta = s_0 (x+ct) - 4 \zeta_1 \zeta_2 \zeta_3 t = s_0(x + c_bt)
\end{equation*}
with $c_b = c - \frac{4 \zeta_1 \zeta_2 \zeta_3}{s_0}$ and $c = 2 (\zeta_1^2 + \zeta_2^2 + \zeta_3^2)$.  The following lemma gives the expression for $s_0$, which is used in (\ref{speed-breather}) and (\ref{localization}).

\begin{lemma}
	\label{lem-kink-2}
Let $\alpha \in (0,K)$ be given in Lemma \ref{lem-5} for $0 < \zeta_3 < \zeta_2 < \zeta_1$. We have 
\begin{equation}
\label{s-expression}
s_0 = \frac{\nu H'(2 \alpha)}{H(2\alpha)} =  \zeta_2 + \zeta_3 - \zeta_1 + \frac{2 \nu \Theta'(\alpha)}{\Theta(\alpha)}.
\end{equation}
\end{lemma}

\begin{proof}
	It follows from $H(x) = \sqrt{k} \sn(x) \Theta(x)$ that 
	\begin{align*}
	\frac{H'(2\alpha)}{H(2\alpha)} = \frac{\Theta'(2\alpha)}{\Theta(2\alpha)} + 
	\frac{\cn(2\alpha) \dn(2\alpha)}{\sn(2\alpha)}.
	\end{align*}
The addition formulas \cite[(3.6.2)]{Lawden}
	\begin{equation}
	\label{Z-addition}
	Z(u\pm v) = Z(u) \pm Z(v) \mp k^2 \sn(u) \sn(v) \sn(u \pm v),
	\end{equation}
imply
	\begin{align*}
	\frac{\Theta'(2\alpha)}{\Theta(2\alpha)} = 2  \frac{\Theta'(\alpha)}{\Theta(\alpha)} - k^2 \sn^2(\alpha) \sn(2\alpha).
	\end{align*}
With the help of (\ref{double-1}) and (\ref{double-2}), we obtain 
	\begin{align*}
	s_0 &= \frac{\nu H'(2 \alpha)}{H(2\alpha)} \\
	&=  \frac{2\nu \Theta'(\alpha)}{\Theta(\alpha)} 
	+\nu \left[ \frac{\cn(2\alpha) \dn(2\alpha)}{\sn(2\alpha)} - k^2 \sn^2(\alpha) \sn(2\alpha) \right] \\
	&= \frac{2\nu \Theta'(\alpha)}{\Theta(\alpha)} + \frac{1}{\zeta_1} \left[ \zeta_2 \zeta_3 - (\zeta_1 - \zeta_2)(\zeta_1 - \zeta_3) \right] \\
	&= \frac{2\nu \Theta'(\alpha)}{\Theta(\alpha)} + \zeta_2 + \zeta_3 - \zeta_1,
	\end{align*}
	which yields (\ref{s-expression}).
\end{proof}

\begin{remark}
The Jacobi's zeta function $Z(x) := \frac{\Theta'(x)}{\Theta(x)}$ is related to the incomplete elliptic integral of the second kind $E(x,k) := \int_0^x \dn^2(x) dx$. Indeed, it follows from Proposition 5 in \cite{HMP23} that 
	$$
	Z'(x) = 1 - k^2 \sn^2(x) - \frac{E(k)}{K(k)} = \dn^2(x) - \frac{E(k)}{K(k)}.
	$$ 
Hence we have $Z(x) = E(x,k) - \frac{E(k)}{K(k)} x$.
\end{remark}

The following lemma gives derivation of the kink breather (\ref{kink-breather}) by using the Darboux transformation (\ref{DT}).

\begin{lemma}
	\label{lem-kink-3}
	Assume $0 < \zeta_3 < \zeta_2< \zeta_1$ and consider the periodic solution of Theorem \ref{theorem-wave} with the elliptic profile $\phi$. The bounded kink breather solution of the mKdV equation (\ref{mkdv}) is given by 
	\begin{equation}
	\label{kink-breather-again}
	u(x,t) = \frac{4 \phi(\xi) \Theta^2(\nu \xi + \alpha) + e^{2(\eta + \eta_0)} \Theta^2(\nu \xi - \alpha) (2 \phi(\xi) \phi'(\xi) - \phi''(\xi) - b)}{4 \Theta^2(\nu \xi + \alpha) + e^{2(\eta + \eta_0)} \Theta^2(\nu \xi - \alpha)  (c + 2 \phi'(\xi) - 2 \phi(\xi)^2)},
	\end{equation}
	where $\eta_0 \in \mathbb{R}$ is the arbitrary translational parameter. 
\end{lemma}

\begin{proof}
The Darboux transformation (\ref{DT}) returns the identity as $\zeta \to 0$ unless $p^2 \to q^2$ as $\zeta \to 0$. Without loss of generality, we take $(p_0,q_0)$ in the form (\ref{p-g-gen}) and expand solutions of the Lax system of linear equations (\ref{LS}) in powers of $i \zeta$:
\begin{equation*}
p = p_0 + i \zeta p_1 + \mathcal{O}(\zeta^2), \qquad 
q = p_0 - i \zeta p_1 + \mathcal{O}(\zeta^2).
\end{equation*}
By using the second-order equation (\ref{second}), we obtain recursively in powers of $i \zeta$:
\begin{align}
\mathcal{O}(1) &: \quad \partial_{\xi} p_0 = \phi(\xi) p_0, &\quad 
\partial_t p_0 = -b p_0, 
\label{order-1} \\
\mathcal{O}(i\zeta) &: \quad
\partial_{\xi} p_1 = -\phi(\xi) p_1 + p_0, &\quad 
\partial_t p_1 = b p_1 + (2\phi(\xi)^2 - 2 \phi'(\xi) - c) p_0.
\label{order-2} 
\end{align}
where $\xi = x + ct$. The solution of the system (\ref{order-1}) at $\mathcal{O}(1)$ agrees with (\ref{p-g-gen}) and yields
\begin{equation}
\label{p-0}
p_0 = e^{\eta} \frac{\Theta(\nu \xi - \alpha)}{\Theta(\nu \xi + \alpha)},
\end{equation}
where $\eta = s_0 \xi - bt = s_0 (x+c_b t)$ with $c_b = c - \frac{b}{s_0}$. Since $p^2 - q^2 = 4 i \zeta p_0 p_1 + \mathcal{O}(\zeta^2)$, the new solution follows from (\ref{DT}) in the limit $\zeta \to 0$ as 
\begin{equation}
\label{kink-breather-0}
\hat{u} = \phi - \frac{p_0}{p_1}.
\end{equation}
By using the systems (\ref{order-1}) and (\ref{order-2}), we derive 
\begin{equation}
\label{system-p0p1}
\partial_{\xi} (p_0 p_1) = p_0^2, \qquad \partial_t (p_0 p_1) = (2 \phi(\xi)^2 - 2 \phi'(\xi) - c) p_0^2.
\end{equation}
Separation of variables $\xi$ and $t$ in the linear system (\ref{system-p0p1}) yields the exact solution
\begin{equation}
\label{p0p1}
p_0 p_1 = \frac{1}{2b} \left[ (c + 2 \phi' - 2 \phi^2) p_0^2 + 4 C \right],
\end{equation}
where $C$ is the constant of integration. Substituting (\ref{p0p1}) into (\ref{kink-breather-0}) and bringing it to the common denominator yields 
\begin{equation}
\label{kink-breather-new}
\hat{u} = \phi - \frac{2b p_0^2}{4C +p_0^2 (c + 2 \phi' - 2 \phi^2)}
= \frac{4C \phi + p_0^2 (2 \phi \phi' - \phi'' - b)}{4C + p_0^2 (c + 2 \phi' - 2 \phi^2)},
\end{equation}
where the second-order equation (\ref{second}) has been used. 

Next we prove that the new solution (\ref{kink-breather-new}) is bounded for $(x,t) \in \R \times \R$ if $C \geq 0$. 
	It follows from (\ref{pcW}), (\ref{philinha}), and (\ref{form-3-squared}) that 
	\begin{align}
	c + 2 \phi'(\xi) - 2 \phi(\xi)^2 =  4 \wp(v) - 4 \wp\left( \xi + \frac{v}{2} \right).
	\label{expression-01}
	\end{align}
By using (\ref{rel-Jac-Wei}), (\ref{rel-par-e-zeta}),  (\ref{v-alpha-correspondence}), and (\ref{shift-by-iK}), we obtain 
	\begin{align*}
\wp\left( \xi + \frac{v}{2} \right) &= e_3 + \frac{e_1 - e_3}{\sn^2(\nu x + \frac{\nu v}{2})} \\
&= \frac{1}{3} (-2\zeta_1^2 + \zeta_2^2 + \zeta_3^2) + \frac{\zeta_1^2 - \zeta_3^2}{\sn^2(\nu x - iK' - \alpha)} \\
&= \frac{1}{3} (-2\zeta_1^2 + \zeta_2^2 + \zeta_3^2) + (\zeta_1^2 - \zeta_2^2) \sn^2(\nu x - \alpha),
\end{align*}
which yields with the help of (\ref{pcW}) that 
	\begin{align}
	c + 2 \phi'(\xi) - 2 \phi(\xi)^2 	&= 4 \zeta_1^2 - 4(\zeta_1^2 - \zeta_2^2) \sn^2(\nu \xi - \alpha) 
	\notag \\
	&= 4 \zeta_1^2 \cn^2(\nu \xi - \alpha) + 4 \zeta_2^2 \sn^2(\nu \xi - \alpha), \label{expression-1}
	\end{align}
Since the expression in (\ref{expression-1}) is strictly positive, 
the denominator in (\ref{kink-breather-new}) is bounded away from zero if $C \geq 0$. Subtituting (\ref{p-0}) and $C = e^{-2 \eta_0} \geq 0$ with arbitrary $\eta_0 \in \R$ in (\ref{kink-breather-new}) yields the expression (\ref{kink-breather-again}), which is bounded for every $(x,t) \in \R \times \R$.
\end{proof}

The following lemma gives the phase shifts of the kink breathers in the limits (\ref{limits-breather}). 

\begin{lemma}
	\label{lem-kink-4}
It follows for the solution (\ref{kink-breather-again}) that 
	\begin{equation}
\label{limits-breather-again}
u(x,t) \to \left\{ \begin{array}{ll} \phi(\xi) \quad & \mbox{\rm as} \;\; \eta \to -\infty, \\
-\phi(\xi - 2 \nu^{-1} \alpha) \quad & \mbox{\rm as} \;\; \eta \to +\infty. 
\end{array} \right.
\end{equation}	
\end{lemma}

\begin{proof}
	It is clear that 
\begin{equation*}
\lim_{\eta \to -\infty} u(x,t) = \phi(\xi) \quad 
\mbox{\rm and} \quad 
\lim_{\eta \to +\infty} u(x,t) = \frac{2 \phi(\xi) \phi'(\xi) - \phi''(\xi) - b}{c + 2 \phi'(\xi) - 2 \phi(\xi)^2},
\end{equation*}	
hence the first limit in (\ref{limits-breather-again}) is confirmed. 
To confirm the second limit in (\ref{limits-breather-again}), 
we obtain from (\ref{pcW}), (\ref{philinha}), and (\ref{form-3-squared}) that 
	\begin{align}
	2 \phi(\xi) \phi'(\xi) - \phi''(\xi) - b =  \frac{d}{d \xi} (\phi(\xi)^2 - \phi'(\xi)) - b = 2 \wp'\left( \xi + \frac{v}{2} \right) - 2 \wp'(v).
		\label{expression-02}
	\end{align}
Combining (\ref{expression-01}) and (\ref{expression-02}) yields
\begin{align*}
	\frac{2\phi(\xi)\phi' (\xi)-\phi''(\xi)-b}{c+2\phi'(\xi)-2\phi(\xi)^2} &= -\frac{1}{2}  \frac{\wp'(\xi+\frac{v}{2}) -\wp'(v)}{\wp(\xi+\frac{v}{2}) -\wp(v)}  = - \phi(\xi+v),
\end{align*}
where we have used (\ref{phinova}) rewritten as 
$$
\phi(\xi) =\frac{1}{2}\frac{\wp'(\xi-\frac{v}{2}) + \wp' (\xi+\frac{v}{2})}{\wp(\xi-\frac{v}{2})-\wp (\xi+\frac{v}{2})} = \frac{1}{2}\frac{\wp' (\xi-\frac{v}{2})-\wp' (v)}{\wp(\xi-\frac{v}{2})-\wp(v)}.
$$
The second equality can be proven based on Exercise 15 in \cite[p.183]{Lawden}, 
$$
\frac{\wp'(u)-\wp'(v)}{\wp(u)-\wp(v)} = \frac{\wp'(v) + \wp'(u+v)}{\wp(v) - \wp(u+v)}, \qquad \forall u,v \in \mathbb{C}.
$$
By using this relation twice together with the even parity of $\wp(x)$, we obtain 
\begin{align*}
\frac{\wp' (\xi-\frac{v}{2})-\wp' (v)}{\wp(\xi-\frac{v}{2})-\wp(v)} &= \frac{\wp' (v) + \wp' (\xi+\frac{v}{2})}{\wp(v)-\wp(\xi+\frac{v}{2})} \\
&= \frac{\wp' (v)- \wp' (-\xi-\frac{v}{2})}{\wp(v)-\wp(-\xi-\frac{v}{2})} \\
&= \frac{\wp' (-\xi-\frac{v}{2}) + \wp' (-\xi+\frac{v}{2})}{\wp(-\xi-\frac{v}{2})-\wp(-\xi+\frac{v}{2})} \\
&= \frac{\wp' (\xi-\frac{v}{2}) + \wp' (\xi+\frac{v}{2})}{\wp(\xi-\frac{v}{2})-\wp(\xi+\frac{v}{2})}
\end{align*}
In view of (\ref{periods}) and (\ref{v-alpha-correspondence}), we get 
$$
\phi(\xi + v) = \phi(\xi - 2 i \nu^{-1} K' - 2 \nu^{-1} \alpha) = \phi(\xi - 2 \nu^{-1} \alpha),
$$
which confirms the second limit in (\ref{limits-breather-again}).
\end{proof}

The proof of Theorem \ref{theorem-breather} is accomplished with the results of 
Lemmas \ref{lem-kink-1}, \ref{lem-kink-2}, \ref{lem-kink-3}, and \ref{lem-kink-4}.

\section{Summary and further discussions}
\label{sec-5}

We have characterized the elliptic profile of the traveling wave solutions of the defocusing mKdV equation by using the elliptic function theory. The representation given by Theorem \ref{theorem-wave} is based on the structure of zeros and poles of the elliptic profile and leads to a simple two-parameter form of Corollary \ref{cor-wave} which incorporates the scaling transformation. 
Based on the new representation of the elliptic profile and the explicit solutions for eigenfunctions of the Lax system for $\zeta = 0$, we have constructed a new solution for the defocusing mKdV equation in Theorem \ref{theorem-breather}, which corresponds to the kink breather propagating over the traveling peiodic wave. When the elliptic profile degenerates into the 
hyperbolic profile, the new solution recovers the two-soliton solution of Corollary \ref{cor-two-solitons}. Overall, our work solves the open problem posed in \cite{MP24} 
on the analytical characterization of the kink breather in the defocusing mKdV equation. 

For a general construction of breathers on the elliptic profile of the traveling wave solution, one needs to obtain the explicit solutions for eigenfunctions of the Lax system with $\zeta \neq 0$. This has been done for the snoidal profile which corresponds to the Riemann's theta function of genus one, see Appendix \ref{app-A}. However, eigenfunctions of the Lax system for the general elliptic profile which corresponds to the Riemann's theta function of genus two have not been characterized uniquely, as we explain below. 

Let $u(x,t) = \phi(x+ct)$ be defined by the elliptic profile $\phi$ of Theorem \ref{theorem-wave} for $0 < \zeta_3 < \zeta_2 < \zeta_1$. Let $\varphi(x,t) = \psi(x+ct) e^{4 \mu t}$ be the eigenfunction of the Lax 
system (\ref{LS}) with the spectral parameter $\zeta \neq 0$. Then, we claim that the eigenfunction $\psi = (p,q)^T$ for $\mu = -i \sqrt{P(\zeta)}$ is given by 
\begin{equation}
\label{p-gen-1}
p(x) = e^{s x} 
\frac{H(\nu x-z_1^*) H(\nu x-z_2^*)}{\Theta(\nu x-\alpha) \Theta(\nu x+ \alpha) \Theta(\alpha + z_1^*) \Theta(\alpha + z_2^*)} e^{\frac{i\pi}{2K} (z_1^* + z_2^*)} 
\end{equation}
and
\begin{equation}
\label{q-gen-1}
q(x) = e^{s x} 
\frac{H(\nu x + z_1) H(\nu x + z_2)}{\Theta(\nu x-\alpha) \Theta(\nu x+ \alpha) \Theta(\alpha - z_1) \Theta(\alpha - z_2)} e^{-\frac{i\pi}{2K} (z_1 + z_2)},
\end{equation}
whereas the eigenfunction $\psi = (p^*,q^*)^T$ for $\mu = i \sqrt{P(\zeta)}$ is given by 
\begin{equation}
\label{p-gen}
p^*(x) = e^{- s x} \frac{H(\nu x-z_1) H(\nu x-z_2)}{\Theta(\nu x-\alpha) \Theta(\nu x+ \alpha) \Theta(\alpha - z_1) \Theta(\alpha - z_2)} e^{-\frac{i\pi}{2K} (z_1 + z_2)} 
\end{equation}
and
\begin{equation}
\label{q-gen}
q^*(x) = -e^{- s x} 
\frac{H(\nu x + z_1^*) H(\nu x + z_2^*)}{\Theta(\nu x-\alpha) \Theta(\nu x+ \alpha) \Theta(\alpha + z_1^*) \Theta(\alpha + z_2^*)} e^{\frac{i\pi}{2K} (z_1^* + z_2^*)}.
\end{equation}
Parameters $\{ \pm z_1, \pm z_2 \}$ and $\{ \pm z_1^*, \pm z_2^* \}$ are the only roots for $z = \nu x$ in $[-K,K] \times [-iK',iK']$ of equations 
\begin{equation}
\label{square-root-1}
\zeta \left[ \phi(x)^2 + 2 \zeta^2 - \zeta_1^2 - \zeta_2^2 - \zeta_3^2 \right] + 2 \sqrt{P(\zeta)} = 0
\end{equation}
and
\begin{equation}
\label{square-root-2}
\zeta \left[ \phi(x)^2 + 2 \zeta^2 - \zeta_1^2 - \zeta_2^2 - \zeta_3^2 \right] - 2 \sqrt{P(\zeta)} = 0,
\end{equation}
respectively, such that $\{ z_1, z_2, z_1^*, z_2^* \}$ are the only roots for 
$z = \nu x$ in $[-K,K] \times [-iK',iK']$ of equation
\begin{equation}
\label{zero-denominator}
\zeta^2 \phi(x) - \frac{i}{2} \zeta \phi'(x) - \zeta_1 \zeta_2 \zeta_3 = 0
\end{equation}
and $\{ -z_1, -z_2, -z_1^*, -z_2^* \}$ are the only roots for 
$z = \nu x$ in $[-K,K] \times [-iK',iK']$ of equation
\begin{equation}
\label{zero-denominator-neg}
\zeta^2 \phi(x) + \frac{i}{2} \zeta \phi'(x) - \zeta_1 \zeta_2 \zeta_3 = 0,
\end{equation}
respectively. We note that (\ref{square-root-2}) and (\ref{zero-denominator}) give zeros of the numerator and denominator of the quotient (\ref{rho1}), respectively, whereas (\ref{square-root-1}) and (\ref{zero-denominator-neg}) give zeros of the denominator and numerator of the quotient (\ref{rho2}), respectively. The roots $\{ z_1, z_2, z_1^*, z_2^* \}$ satisfy the completeness relation
\begin{equation}
\label{root-completeness}
z_1 + z_2 + z_1^* + z_2^* = 0 \; \mbox{\rm mod } (2K, 2iK'),
\end{equation}
whereas the expression for $s$ is given by 
\begin{align}
s = -i \frac{\zeta^4 - \zeta^2 (\zeta_2 + \zeta_3 - \zeta_1)^2 
	+ \zeta \sqrt{P(\zeta)} + \zeta_1 \zeta_2 \zeta_3 (\zeta_2 + \zeta_3 - \zeta_1)}{\zeta^3 + \zeta (\zeta_2 \zeta_3 - \zeta_1 \zeta_2 - \zeta_1 \zeta_3) + \sqrt{P(\zeta)}} 
	- \frac{\nu H'(z_1)}{H(z_1)} - \frac{\nu H'(z_2)}{H(z_2)}.
	\label{s-final}
\end{align}
To complete the characterization of the eigenfunctions (\ref{p-gen-1})--(\ref{q-gen-1}) and (\ref{p-gen})--(\ref{q-gen}), one needs to 
find uniquely expressions for $\{z_1,z_2,z_1^*,z_2^*\}$ in terms of $\zeta \in \mathbb{R}$. We checked numerically in the gaps of the Lax spectrum for $\zeta \in (\zeta_2,\zeta_1)$ and $\zeta \in (0,\zeta_3)$, see Figure \ref{fig-1} (right), that the roots satify the symmetry $z_1 = -\bar{z}^*_1$ and $z_2 = -\bar{z}^*_2$ for these values of $\zeta$. With the account of the completeness relation (\ref{root-completeness}) and the symmetry $z_1 = -\bar{z}^*_1$ and $z_2 = -\bar{z}^*_2$, there are still three real parameters in the roots $\{z_1,z_2,z_1^*,z_2^*\}$, which must be uniquely defined in terms of the only spectral parameter $\zeta \in \mathbb{R}$.

The next example shows that the general construction of the eigenfunctions reduces as $\zeta \to 0$ to the explicit solutions obtained in Section \ref{sec-4}. 

\begin{example}
	\label{example-zero-zeta}
	If $\zeta \to 0$, it follows from (\ref{square-root-1}), (\ref{square-root-2}), (\ref{zero-denominator}), and (\ref{zero-denominator-neg}) that both pairs of the roots $\{z_1,z_2\}$ and $\{ z_1^*,z_2^*\}$ are defined from the poles of $\phi(x)$ for $z = \nu x$ at $\pm (iK' + \alpha)$ and $\pm (iK' - \alpha)$ as follows:
	$$
	z_1 = i K' - \alpha, \quad	z_2 = -i K' - \alpha, \quad  
	z_1^* = iK' + \alpha, \quad z_2^* = -iK' + \alpha.
	$$
The four roots satisfy the completeness relation (\ref{root-completeness}) 
and the symmetry $z_1 = -\bar{z}^*_1$ and $z_2 = -\bar{z}^*_2$.
It follows from (\ref{p-gen-1})--(\ref{q-gen-1}) by using (\ref{H-Theta-connection}) that  
$$
p(x) = q(x) = C e^{s_0 x} \frac{\Theta(\nu x - \alpha)}{\Theta(\nu x + \alpha)},
$$
with the numerical constant 
$$
C = \frac{e^{\frac{\pi K'}{2K} + \frac{i \pi \alpha}{K}}}{\Theta(2\alpha+ i K') \Theta (2\alpha - iK')}.
$$
The constant $s_0$ is defined from  (\ref{s-final}) as follows:
\begin{align*}
s_0 &= \zeta_2 + \zeta_3 - \zeta_1 - \frac{\nu H'(iK' - \alpha)}{H(iK' -\alpha)} - \frac{\nu H'(-iK' - \alpha)}{H(-iK' -\alpha)} \\
&= \zeta_2 + \zeta_3 - \zeta_1 + \frac{2\nu \Theta'(\alpha)}{\Theta(\alpha)}.
\end{align*}
This construction agrees with (\ref{p-g-gen}) and (\ref{localization}). Similarly, we obtain from (\ref{p-gen})--(\ref{q-gen}) that 
	$$
p(x) = -q(x) = C e^{-s_0 x} \frac{\Theta(\nu x + \alpha)}{\Theta(\nu x - \alpha)},
	$$
which recovers (\ref{p-g-gen-second}) with the same $s_0$ and the same constant $C$.
\end{example}

We conclude from Example \ref{example-zero-zeta} that the characterization of eigenfunctions of the Lax system for the general traveling wave with the elliptic profile $\phi$ is recovered correctly in 
the limit $\zeta \to 0$. However, it needs to be completed with the 
explicit mapping $\zeta \mapsto z_1,z_2,z_1^*,z_2^*$ for $\zeta \neq 0$ before the expressions (\ref{p-gen-1})--(\ref{q-gen-1}) and (\ref{p-gen})--(\ref{q-gen}) can be used for the construction of general breathers on the general traveling wave. This unique characterization is an open question for further studies.

\appendix

\section{Derivation of the explicit eigenfunctions for the snoidal wave} 
\label{app-A}

Let $u(x,t) = \phi(x+ct)$ with $\phi(x) = k \sn(x)$ and $c = 1 + k^2$ be the 
periodic solution of the mKdV equation (\ref{mkdv}) in Example \ref{ex-snoidal}. 
Let $\varphi(x,t) = \psi(x+ct) e^{4 \mu t}$ be the eigenfunction of the Lax 
system (\ref{LS}). It follows from (\ref{spectral}), (\ref{time}), (\ref{char-eq}), (\ref{char-poly}), and (\ref{factor-P}) with $\zeta_3 = 0$ that 
$\psi(x)$ is the eigenfunction of the spectral problem 
\begin{equation}
\frac{d}{dx} \psi = \left( \begin{matrix} i \zeta & k \sn(x) \\ k \sn(x) & -i \zeta \end{matrix} \right) \psi
\label{spectral-app}
\end{equation}
and $\mu = \pm i \sqrt{P(\zeta)}$, where $P(\zeta) = \zeta^2 (\zeta^2 - \zeta_1^2) (\zeta^2 - \zeta_2^2)$. Eigenfunctions of the spectral problem 
(\ref{spectral-app}) are available in the explicit form, see \cite{B,Takahashi}, compared to the open problem posed in Section \ref{sec-5}. The following proposition reviews details of the derivation of the explicit 
solution of the spectral problem (\ref{spectral-app}). 
We give it for completeness as the explicit expression has been 
used in \cite{MP24} without verification of its validity. 

\begin{proposition}
	\label{prop-app}
	Let $\zeta_3 = 0$, $\zeta_2 = \frac{1}{2} (1-k)$, and $\zeta_1 = \frac{1}{2} (1+k)$, where $k \in (0,1)$ is the elliptic modulus of the snoidal solution $\phi(x) = k \sn(x)$ in Example \ref{ex-snoidal}. Define 
	$z \in [-K,K] \times [-iK',iK']$ from the spectral parameter $\zeta \in \mathbb{R}$ by using the characteristic relation 
	\begin{equation}
	\label{zeta-definition}
	\zeta(z) = \frac{1}{2}\dn(z)\dn(iK' - z). 
	\end{equation}
One solution $\psi = (p,q)^T$ of the spectral problem (\ref{spectral-app}) 
with $\mu = i \sqrt{P(\zeta)}$ is given by 
\begin{equation}
\label{p-q}
p(x) = e^{s(z) x}  \frac{H(x-z)}{\Theta(x) \Theta(z)}, \quad 
q(x,t) = e^{s(z) x}  \frac{\Theta(x-z)}{\Theta(x) H(z)},
\end{equation}
where 
\begin{align}
\label{s}
s(z) &= \frac{\Theta'(z)}{2 \Theta(z)} - \frac{\Theta'(iK'-z)}{2 \Theta(iK' - z)} - \frac{i \pi}{4K}, \\
\label{omega}
\mu(z) &= \frac{i}{8} \dn(z) \dn(iK' - z) \left[ \frac{1}{\sn^2(z)} - \frac{1}{\sn^2(iK'-z)} \right],
\end{align}
Another solution $\psi = (p^*,q^*)^T$ of the spectral problem (\ref{spectral-app}) with $\mu = -i \sqrt{P(\zeta)}$  is obtained from (\ref{p-q}) by replacing $z$ by $z' := iK' - z$. 
\end{proposition}

\begin{proof}
Recall that the quotient $\rho = p/q$ for the solution $\psi = (p,q)^T$ 
with $\mu = i \sqrt{P(\zeta)}$ in (\ref{mu}) is defined by either (\ref{rho1}) or (\ref{rho2}). Substituting $\zeta_3 = 0$ and canceling one power of $\zeta$ in  (\ref{rho1}) yields
\begin{equation}
\label{rho-snoidal}
\rho = -i \frac{\zeta^2 + \frac{1}{2} (\phi^2 - \zeta_1^2 - \zeta_2^2) - \sqrt{(\zeta^2 - \zeta_1^2)(\zeta^2-\zeta_2^2)}}{\zeta \phi - \frac{i}{2} \phi'}.
\end{equation}
By Proposition \ref{prop-1}, there exist exactly two zeros in $[-K,K] \times [-iK',iK']$ of the numerator in (\ref{rho-snoidal}) since $\phi(x) = k \sn(x)$. Similarly, there exist exactly two zeros in $[-K,K] \times [-iK',iK']$ of the denominator in (\ref{rho-snoidal}) rewritten as 
\begin{align}
\label{root-denom}
\frac{\phi'(x)}{\phi(x)} = \frac{\cn(x) \dn(x)}{\sn(x)} = -2i \zeta \in \mathbb{C}.
\end{align}
This also follows from Proposition \ref{prop-1} due to the fundamental relations for elliptic functions (\ref{fund-elliptic}). For a given $\zeta \in \mathbb{R}$, we denote one root of (\ref{root-denom}) by $z \in  [-K,K] \times [-iK',iK']$. With the help of (\ref{shift-by-iK}), this allows us to parameterize 
\begin{equation}
\label{zeta-param-1}
\zeta(z) = \frac{i}{2} \frac{\cn(z) \dn(z)}{\sn(z)} = \frac{1}{2} \dn(z) \dn(iK' - z),
\end{equation}
which recovers (\ref{zeta-definition}). Due to the symmetry, the other root of (\ref{root-denom}) is $iK'-z \in  [-K,K] \times [-iK',iK']$. Thus, $z, iK'-z$ are two roots of the denominator in (\ref{rho-snoidal}) in  $[-K,K] \times [-iK',iK']$.

To get roots of the numerator in (\ref{rho-snoidal}), we verify that 
\begin{align*}
(\zeta^2 - \zeta_1^2) (\zeta^2 - \zeta_2^2) &= \zeta^4 - \frac{1}{2} (1+k^2) \zeta^2 + \frac{1}{16} (1-k^2)^2 \\
&= \frac{\cn^4(z) \dn^4(z) + 2 (1+k^2) \sn^2(z) \cn^2(z) \dn^2(z) + (1-k^2)^2 \sn^4(z)}{16 \sn^4(z)} \\
&= \frac{(1 - k^2 \sn^4(z))^2}{16 \sn^4(z)}.
\end{align*}
This yields with the help of (\ref{shift-by-iK}) that 
\begin{align*}
\mu(z) &= i \zeta \sqrt{(\zeta^2 - \zeta_1^2) (\zeta^2 - \zeta_2^2)} \\
&= \frac{i}{8} \dn(z) \dn(iK' - z) \left[ \frac{1}{\sn^2(z)} - \frac{1}{\sn^2(iK'-z)} \right],
\end{align*}
which recovers (\ref{omega}). Moreover, we obtain 
\begin{align*}
& \quad \zeta^2 + \frac{1}{2} \left[ \phi(iK' \pm z)^2 - \zeta_1^2 - \zeta_2^2 \right] - \sqrt{(\zeta^2 - \zeta_1^2)(\zeta^2-\zeta_2^2)} \\
&= -\frac{\cn^2(z) \dn^2(z) - 1 + (1+k^2) \sn^2(z) - k^2 \sn^4(z)}{4 \sn^2(z)} \\
&= 0.
\end{align*}
Hence, $\pm (i K' - z) \in [-K,K] \times [-iK',iK']$ are two zeros of the numerator in (\ref{rho-snoidal}). 

Poles of both numerator and denominator in (\ref{rho-snoidal}) coincide with the poles of $\phi^2(x)$ and $\phi'(x)$, which are double poles at the 
same location $iK'$. Since both the numerator and the denominator 
are elliptic functions, have only two zeros in $[-K,K] \times [-iK',iK']$, and their double poles coincide at $i K'$, Proposition \ref{prop-3} implies that 
\begin{equation}
\label{rho-appendix}
\rho(x) = C_1 e^{-\frac{i \pi x}{2K}} \frac{H(x - iK' - z) \cancel{H(x - iK' + z)}}{H(x-z) \; \cancel{H(x-iK' + z)}} = C_2 \frac{\Theta(x-z)}{H(x-z)}, 
\end{equation}
where we have used (\ref{H-Theta-connection}) and introduced constants 
$C_1$ and 
$C_2 = (-i) C_1 e^{\frac{\pi K'}{4K}} e^{\frac{-i \pi z}{2K}}$.
The exponential factor $e^{-\frac{i \pi x}{2K}}$ in (\ref{rho-appendix}) is included due to the anti-periodicity of $\phi(x) = k \sn(x)$ with respect to the period $2K$ of the elliptic function $\phi^2$. The constant $C_2$ can be uniquely obtained 
from $x = 0$ as 
\begin{align*}
C_2 &= -\frac{H(z)}{\Theta(z)} \rho(0) \\
&= -\frac{2 H(z)}{k \Theta(z)} \left( \zeta^2 - \frac{1}{4} (1+k^2) - \frac{1 - k^2 \sn^2(z)}{4 \sn^2(z)} \right) \\
&= \frac{H(z)}{k \sn^2(z) \Theta(z)} \\
&= \frac{\Theta(z)}{H(z)},
\end{align*}
where we have used $H(z) = \sqrt{k} \sn(z) \Theta(z)$. 
This recovers the quotient in (\ref{p-q}) as 
\begin{equation}
\label{rho-paramet}
\rho(x) = \frac{ \Theta(z) \Theta(x-z)}{H(z) H(x-z)}.
\end{equation}

It remains to integrate (\ref{p-eq}) for $p(x)$ in order to verify the expression (\ref{s}) for $s(z)$ and the representation (\ref{p-q}) for $p$. 
We rewrite (\ref{p-eq}) with (\ref{zeta-param-1}) and (\ref{rho-paramet}) as
\begin{equation}
\label{p-parameter}
p'(x) = \left( -\frac{\cn(z) \dn(z)}{2 \sn(z)} + \frac{\sn(x)}{\sn(z) \sn(x-z)} \right) p(x),
\end{equation}
where we have used $H(x) = \sqrt{k} \sn(x) \Theta(x)$. Since 
$$
\Theta(z + iK') = i e^{\frac{\pi K'}{4K}} e^{-\frac{i \pi z}{2K}} H(z),
$$
we obtain 
\begin{align*}
\frac{\Theta'(z)}{2 \Theta(z)} - \frac{\Theta'(iK - z)}{2 \Theta(iK' - z)} - \frac{i \pi}{4K} = \frac{\Theta'(z)}{2 \Theta(z)} + \frac{H'(z)}{2 H(z)} = -\frac{\cn(z) \dn(z)}{2 \sn(z)} + \frac{H'(z)}{H(z)}.
\end{align*}
Hence (\ref{p-parameter}) is satisfied with 
$$
p(x) = e^{s(z) x} \frac{H(x-z)}{\Theta(x) \Theta(z)},
$$
where $s(z)$ is given by (\ref{s}) if and only if the following identity is true:
\begin{equation*}
\frac{H'(x-z)}{H(x-z)} - \frac{\Theta'(x)}{\Theta(x)} = \frac{\sn(x)}{\sn(z) \sn(x-z)} - \frac{H'(z)}{H(z)}.
\end{equation*}
This identity can be rewritten with the help of $H(x) = \sqrt{k} \sn(x) \Theta(x)$ in the equivalent form:
\begin{align*}
\frac{\Theta'(x-z)}{\Theta(x-z)} - \frac{\Theta'(x)}{\Theta(x)} + \frac{\Theta'(z)}{\Theta(z)} &= \frac{\sn(x)}{\sn(z) \sn(x-z)} - \frac{\cn(x-z) \dn(x-z)}{\sn(x-z)} - \frac{\cn(z) \dn(z)}{\sn(z)} \\
&= \frac{\sn(x) - \sn(x-z) \cn(z) \dn(z) - \sn(z) \cn(x-z) \dn(x-z)}{\sn(z) \sn(x-z)} \\
&= k^2 \sn(x) \sn(z) \sn(x-z),
\end{align*}
where we have used (\ref{sn-addition}). This relation 
is satisfied due to the addition formulas (\ref{Z-addition}).
The proof of (\ref{zeta-definition}), (\ref{p-q}), (\ref{s}), and (\ref{omega}) is complete. The other solution for $\mu = -i \sqrt{P(\zeta)}$ is obtained from  (\ref{zeta-definition}), (\ref{p-q}), (\ref{s}), and (\ref{omega}) by replacing $z$ with $z' := i K' - z$, which is another root of (\ref{root-denom}) in  $[-K,K] \times [-iK',iK']$.
\end{proof}

\begin{remark}
	The characteristic equation (\ref{zeta-definition}) is analyzed in \cite[Proposition 1]{MP24} with precise behavior of $z \in [-K,K] \times [-iK',iK']$ for $\zeta \in \R$ shown in Figure 5 in \cite{MP24}. In particular, if $\zeta \in (\zeta_2,\zeta_1)$ is in the bandgap of the Lax spectrum, see Figure \ref{fig-3} (left), then $z - \frac{i K'}{2} \in (0,K)$.
\end{remark}

{\bf Acknowledgment:} The first author is supported by the S\~ao Paulo Research Foundation (FAPESP), Brazil, Process Number 2022/14833-3. The second author 
is supported by the NSERC Discovery grant. This study was conducted during the research visit of L. K. Arruda at McMaster University (2023-2024) and during the stay of D. E. Pelinovsky at Newcastle, UK in the Isaac Newton Institute Programme ``Emergent phenomena in nonlinear dispersive waves" (July-August, 2024).

\end{document}